\documentclass[twoside]{article}

\usepackage[accepted]{aistats2022}
\usepackage{xcolor}
\usepackage{tikz}
\usepackage[utf8]{inputenc}
\usepackage{xr}
\usepackage[style = apa, natbib = true]{biblatex}
\usepackage{xcolor} 
\usepackage{amsmath}
\usepackage{hyperref}
\usepackage{amssymb}
\usepackage[acronym]{glossaries}
\usepackage{stmaryrd}
\usepackage{dsfont}
\usepackage{graphicx}
\usepackage[ruled]{algorithm2e}
\usepackage{amsthm}
\newtheorem{theorem}{Theorem}
\newacronym{dpmm}{DPMM}{Dirichlet process mixture model}
\newacronym{ntr}{NTR}{neutral-to-the-right}
\newacronym{iid}{i.i.d.}{independently and identically distributed}
\newacronym{pbc}{PBC}{primary biliary cirrhosis}
\newacronym{IS}{IS}{importance sampling}
\newacronym{smc}{SMC}{sequential Monte Carlo}
\newacronym{ess}{ESS}{effective sample size}
\newacronym{aft}{AFT}{accelerated failure time}
\newacronym{mcmc}{MCMC}{Markov chain Monte Carlo}
\newacronym{cdf}{CDF}{cumulative distribution function}
\newacronym{km}{KM}{Kaplan-Meier}
\newacronym{mle}{MLE}{maximum likelihood estimate}
\newacronym{rv}{r.v.}{random variable}

\DeclareMathOperator*{\argmin}{arg\,min}

\newcommand{\iid}{\overset{\mathrm{iid}}{\sim}}

\DeclareRobustCommand\full {\tikz[baseline=-0.5ex]\draw[thick,color=red!60!black] (0,0)--(0.5,0);}
\DeclareRobustCommand\fullmid  {\tikz[baseline=-0.5ex]\draw[thick,color=red!60!black,opacity = 0.5] (0,0)--(0.5,0);}
\DeclareRobustCommand\fulllow  {\tikz[baseline=-0.5ex]\draw[thick,color=red!60!black, opacity = 0.25] (0,0)--(0.5,0);}

\DeclareRobustCommand\fullmidk  {\tikz[baseline=-0.5ex]\draw[thick,opacity = 0.6] (0,0)--(0.5,0);}
\DeclareRobustCommand\dottedmidk{\tikz[baseline=-0.5ex]\draw[thick,dotted,opacity = 1.] (0.06,0)--(0.5,0);}

\begin{document}

%

%

\twocolumn[

\aistatstitle{A Predictive Approach to Bayesian Nonparametric Survival Analysis}

\aistatsauthor{ Edwin Fong \And Brieuc Lehmann}

\aistatsaddress{University of Oxford  \And University College London} ]

\begin{abstract}
  Bayesian nonparametric methods are a popular choice for analysing survival data due to their ability to flexibly model the distribution of survival times. These methods typically employ a nonparametric prior on the survival function that is conjugate with respect to right-censored data. Eliciting these priors, particularly in the presence of covariates, can be challenging and inference typically relies on computationally intensive Markov chain Monte Carlo schemes. In this paper, we build on recent work that recasts Bayesian inference as assigning a predictive distribution on the unseen values of a population conditional on the observed samples, thus avoiding the need to specify a complex prior. We describe a copula-based predictive update which admits a scalable sequential importance sampling algorithm to perform inference that properly accounts for right-censoring. We provide theoretical justification through an extension of Doob's consistency theorem and illustrate the method on a number of simulated and real data sets, including an example with covariates. Our approach enables analysts to perform Bayesian nonparametric inference through only the specification of a predictive distribution.
\end{abstract}

\section{INTRODUCTION}

Survival data, also known as time-to-event data, is ubiquitous in a number of domains including economics, engineering, biology, and medicine. {Common examples include the time to failure of a mechanical component, or the time to death of an individual following treatment. The overarching aim of survival analysis is to study the distribution of these survival times. In survival regression, the aim is to assess the effect of covariates on survival time.} 

{A characteristic feature of survival data is that it is often censored - that is, we may not know the survival time exactly. In the case of \textit{right-censoring}, we only observe the information $Y > c$, where $Y$ is the time-to-event of interest and $c$ is the observed censoring time. Right-censoring can occur, for example, if a subject leaves a study before the event of interest has occurred. The partial nature of the information associated with the observed data poses some challenges to statistical inference.} 

A primary goal in survival analysis is to \textit{predict} the survival time for a new individual, perhaps taking into account known covariates (e.g. age) for said individual. In other words, the aim is to learn a predictive distribution $p(y_{n+1} | x_{n+1}, \{x_i, y_i\}_{i = 1:n})$, where $\{x_i, y_i\}_{i = 1:n}$ is an observed training set. To reduce notational burden, we henceforth omit reference to covariates $x$. The standard Bayesian approach to this problem is to first specify a data-generating distribution $f_\theta(y)$, depending on a (potentially infinite-dimensional) parameter $\theta$, and prior $\pi(\theta)$. The predictive distribution is then taken to be the posterior predictive distribution. In the uncensored case, this is
\begin{equation} \label{eq:Bayes_pred}
        p(y_{n+1} | y_{1:n}) = \int f_\theta(y_{n+1}) \,\pi(\theta | y_{1:n})\,d\theta,
\end{equation}
where $\pi(\theta | y_{1:n}) \propto \pi(\theta)\prod_{i=1}^n f_\theta(y_i)$ is the posterior distribution, which is often also of interest.

Here, we take a more direct approach to prediction and posterior inference by explicitly specifying a predictive distribution instead of the usual likelihood and prior.  {In particular, we extend the notion of \textit{martingale posterior distributions} \citep{Fong2021} to right-censored data, appropriately accounting for the partially observed nature of the censored values. In doing so, we leverage one of the key advantages of the martingale posterior framework in replacing the standard \gls*{mcmc} approach to posterior computation with a GPU-friendly and parallelisable optimisation-based algorithm}. Our main contributions are as follows: a) we describe a class of copula-based predictive updates that are suitable under right-censoring; b) we extend Doob's consistency theorem to the setting with right-censored observations, confirming the conceptual equivalence of standard Bayesian inference and the martingale posterior in this setting; c) to perform inference, we develop a sequential \gls*{IS} procedure, avoiding the need for more computationally intensive \gls*{mcmc} algorithms.

\section{RELATED WORK}

There is a rich history of Bayesian nonparametric methods for the analysis of survival data. These typically employ a \gls*{ntr} process \citep{doksum1974} prior on the survival function, chosen for its conjugacy property with respect to censored data \citep{ferguson1979}. {Some examples of such priors include the extended gamma process \citep{kalbfleisch1978}, the beta process \citep{hjort1990}, and the beta-Stacy process \citep{muliere1997a}.  \citet{muliere1997b} offered a generalisation of the beta process based on a P\'{o}lya tree prior.} 
Yet another alternative approach was taken by \citet{kottas2006}, who modelled the {distribution of survival times} using a \gls*{dpmm} with a Weibull kernel. Our copula-based predictive update is intimately linked to the \gls*{dpmm} (see Section~\ref{sec:bivariate_copula}). 

Building on these foundations, extensions to survival regression have been developed based on proportional hazard models, for example by \citet{kalbfleisch1978,hjort1990,kim2003a}. \citet{RivaPalacio2021} relax the restriction of proportionality through the use of a vector of completely random measures. \citet{deIorio2009} developed a dependent Dirichlet process (DDP) mixture model for survival regression that also permits survival curves to cross in the context of a treatment effect analysis. Further examples can be found in \citet[Chapter~13]{ghosal2017}.

The idea of focusing directly on the specification of a predictive distribution goes back to at least \citet{hill1968}, who posited a uniform distribution on the intervals between the order statistics of the observations. Extensions of Hill's predictive distribution to censored data have been proposed by \citet{berliner1988} and \citet{coolen2004}. We build on recent work that proposes to relax the assumption of exchangeability in favour of conditionally identically distributed \citep{Berti2004} sequences, thus allowing for more flexible specifications of the predictive distribution \citep{berti2021}. In particular, we focus on one-step-ahead predictive updates based on bivariate copulas, initially proposed in \citet{Hahn2018} for the uncensored case. As noted in \citet{Fong2021}, there are also connections between this predictive approach and the Bayesian bootstrap \citep{rubin1981} and its extensions to censored data \citep{lo1993, arfe2020}.

\section{BACKGROUND}
In this section, we describe the martingale posterior distribution framework in the \textit{uncensored} \gls*{iid} data setting, as introduced by \citet{Fong2021}. In this work, Bayesian inference is reframed as an imputation problem, where the task is to elicit the joint predictive density on the missing information, which is the remainder of the population $y_{n+1:\infty}$ given an observed sample $y_{1:n}$ in the \gls{iid} case. The joint density of interest can be written as a product of 1-step-ahead predictive densities,
\begin{equation}\label{eq:chain_rule}
p(y_{n+1:N} \mid y_{1:n}) = \prod_{i=n+1}^N p_{i-1}(y_i),
\end{equation}
where we write $p_i(y) := p(y \mid y_{1:i})$ with corresponding \gls*{cdf} $P_i(y)$. Intuitively, a general statistic of interest can then be recovered from the population $y_{1:N}$, which is written as $\theta(y_{1:N})$. The predictive uncertainty in $y_{n+1:N}$ then induces a distribution on $\theta(y_{1:N})$. We will formalize these notions later on.

For the parametric Bayesian with{ sampling density} $f_\theta(y)$ and prior $\pi(\theta)$, the posterior predictive density $p_i(y)$ is defined
as in (\ref{eq:Bayes_pred}). The statistic is then an estimate of $\theta$ indexing the sampling density, e.g. the posterior mean, $\bar{\theta}_N = E\left[\Theta \mid y_{1:N}\right]$,{ where $\Theta$ is the Bayesian random parameter that is marginally distributed according to the prior $\pi$}. With this choice, it can be shown through Doob's consistency theorem \citep{Doob1949} that the above scheme is equivalent to posterior sampling in the limit of $N\to \infty$, that is $\bar{\theta}_\infty  \sim \pi(\theta \mid y_{1:n})$, { where $\bar{\theta}_\infty := \lim_{N\to \infty}\bar{\theta}_N$}. Through this result, parameters are viewed as functions of 
the population of observables, and Bayesian uncertainty can intuitively be seen to arise from subjective uncertainty on the missing remainder of the population. 

\subsection{Martingale Posterior Distributions}
The martingale posterior distribution considers more general sequences of predictive distributions than that induced by the likelihood and prior, and is hence a generalisation of standard Bayesian inference. Given a sequence of predictive \gls*{cdf}s $P_n(y),P_{n+1}(y),\ldots$, one can impute the remainder of the infinite population through the scheme
$$
Y_{n+1} \sim P_n(y),\quad  Y_{n+2}\sim P_{n+1}(y), \quad\ldots ~.
$$
This sequential imputation scheme is termed \textit{predictive resampling}, as it is inspired by the Pólya urn scheme  \citep{blackwell1973} for the Bayesian bootstrap. In practice, it is infeasible to work with infinite populations, so we terminate predictive resampling at $Y_N$ for some large $N > n$.  However, there are still constraints on this sequence $P_i$ needed to ensure a notion of predictive coherence, and in particular for the random limiting empirical distribution to exist {so that we can compute a functional of interest}. The limiting empirical distribution is given by
$$
F_\infty(y) = \lim_{N\to \infty}\frac{1}{N}\left\{\sum_{i=1}^n \mathds{1}(y_i \leq y) + \sum_{i=n+1}^N \mathds{1}(Y_i \leq y) \right\}.
$$
A sufficient condition{ for the existence of $F_\infty$} is that the sequence $P_i$ implies a \textit{conditionally identically distributed} (c.i.d.) sequence of \glspl*{rv}, as investigated in \citet{Berti2004}. The sequence $Y_{n+1},Y_{n+2},\ldots$ is c.i.d. if 
$$
P(Y_{i+k}\leq y \mid y_{1:i}) = P_i(y), \quad \forall k > 0.
$$
This ensures that the sequence of predictive distributions $P_i$ is a martingale, and that predictive resampling returns a well-defined empirical distribution. 

Moreover, the parameter of interest no longer needs to index a sampling density. Instead, it can be defined as 
$$
\theta_0 := \theta(F_0) = \argmin_\theta \int \ell(\theta,y) \, dF_0(y)
$$
where $F_0$ is the true sampling density and the loss function $\ell(\theta,y)$ is elicited by the analyst. After predictive resampling, a sample from the martingale posterior can be recovered by computing $\theta_\infty = \theta(F_\infty)$.  

\subsection{Bivariate Copula Updates} \label{sec:bivariate_copula}
We now discuss a concrete example of a c.i.d. sequence of predictive densities that is both computationally feasible and no longer relies on the likelihood and prior. A useful class of predictive densities depends on the \textit{bivariate copula density}, and was introduced in \citet{Hahn2018}. Briefly, the bivariate copula density is a bivariate probability density function with uniform marginals, that is $d: [0,1]^2 \to \mathbb{R}$ where $\int d(u,v) \, du = \int d(u,v) \, dv = 1$. See \citet{Nelsen2007} for more details. For univariate data, a sequence of predictive densities can be defined recursively through
\begin{equation}\label{eq:copula_seq}
p_i(y) = d_i\{P_{i-1}(y),P_{i-1}(y_i)\}\, p_{i-1}(y).
\end{equation}
Here, $d_i$ is a sequence of bivariate copula densities which models the dependence between $Y_i$ and $Y_{i+1}$. \citet{Hahn2018} showed that all Bayesian predictives have an update of this form, although it is usually intractable for non-conjugate models. A tractable sequence of copula densities is then introduced, inspired by the \gls*{dpmm}, which does not correspond to a Bayesian model. Exploiting the c.i.d. property of this update, \citet{Fong2021} explore their use in the context of martingale posteriors  and provide further extensions to multivariate data and regression.

\subsubsection{Bivariate Copula Updates on $\mathbb{R}^+$}

The updates introduced in \citet{Hahn2018} are applicable to data with support on the entire real line $\mathbb{R}$, and are motivated by the  \gls*{dpmm} with the normal kernel. Survival times, however, are typically strictly positive, so we will now introduce a copula update for data supported on the positive reals, $\mathbb{R}^+$.

We begin by introducing said copula update in the absence of censoring, which is motivated by the first posterior predictive update of the \gls*{dpmm} with an exponential kernel. The \gls*{dpmm} can be written as
\begin{equation*}
\begin{aligned}
     &~f_G(y) = \int \text{Exp}(y \mid \theta) \, dG(\theta) \\ G &\sim \text{DP}(c,G_0), \quad G_0 = \Gamma(\theta \mid a,1).
     \end{aligned}
\end{equation*}
where $\text{Exp}(y \mid \theta)$ is the exponential density with rate $\theta$, and $\Gamma(\cdot \mid a,1)$ is the gamma density with shape $a$ and rate $1$. In Appendix \ref{sec:supp_meth}, we show that this inspires the update
\begin{equation}\label{eq:positive_update}
\begin{aligned}
    p_{i}(y) &= \left[ 1-\alpha_{i} + \alpha_{i}d_{a}\left\{P_{i-1}(y),P_{i-1}(y_{i}) \right\}\right] p_{i-1}(y)\\
d_a(u,v) &= \frac{a + 1}{a} \frac{(1-u)^{-\frac{a + 1}{a}}(1-v)^{-\frac{a + 1}{a}}}{\left\{(1-u)^{-\frac{1}{a}} + (1-v)^{-\frac{1}{a}} -1	\right\}^{a + 2}}\cdot
\end{aligned}
\end{equation}
In fact, the above update corresponds exactly to the \gls*{dpmm} update from $p_0 \to p_1$, but is different for $p_i$ with $i > 1$. The sequence $\alpha_i$ should in general be $\mathcal{O}(i^{-1})$ to approach the independent copula for consistent estimation, and the  specific suggestion of 
$\alpha_i = \left(2-{1}/{i}\right) /({i+1})
$ is motivated in \citet{Fong2021}. Note that the above is a mixture of the independence copula and the Clayton copula \citep{Clayton1978}, as was also pointed out in \citet{Hahn2018}. See \citet[Chapter 2.9]{Balakrishnan2009} for more details. The update for the \gls*{cdf} $P_i(y)$ is similarly tractable and is derived in Appendix \ref{sec:supp_meth}. 

Here, $a > 0$ acts as a bandwidth term, where smaller values indicates a stronger peak; the update in (\ref{eq:positive_update}) is analogous to a kernel density estimate but on $\mathbb{R}^+$. This is illustrated in Figure \ref{fig:copula_bw}a in which we plot the copula kernel $d_a(u_{i-1},v_{i-1})\, p_{i-1}(y)$ for decreasing values of $a$, where $u_{i-1} = P_{i-1}(y), v_{i-1} = P_{i-1}(y_{i})$. The updated density is a weighted mixture of $p_i$ (dashed) and the copula kernel (solid), which is shown in  Figure \ref{fig:copula_bw}b. 

 \begin{figure*}[!ht]
\centering
 \includegraphics[width=0.9\textwidth]{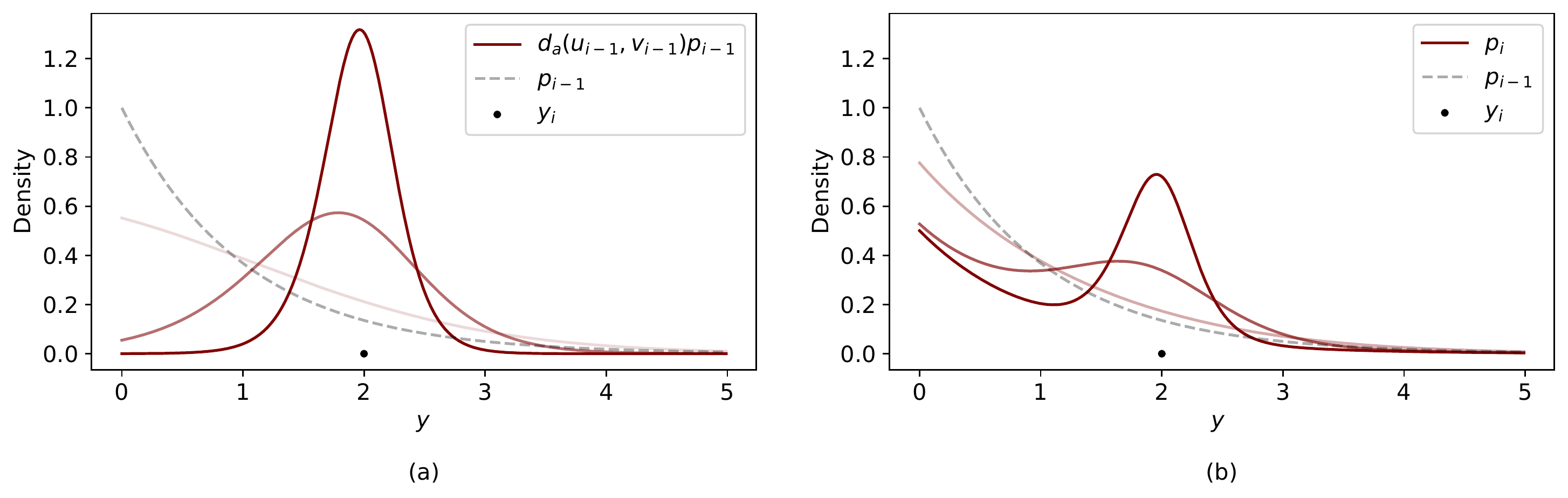}
 \vspace{-5mm}
  \caption{Plot of the (a) copula kernel $d_a(u_{i-1},v_{i-1})\, p_{i-1}(y)$ and (b) updated predictive $p_{i}(y)$, for $a =2,0.5,0.2$ (\fulllow,\fullmid,\full).}\label{fig:copula_bw} \vspace{-2mm}
\end{figure*}

In the case of regression with covariates $\mathbf{x} \in \mathbb{R}^d$, a similar argument based on the DDP mixture model can be used to derive an update for the conditional density $p_i(y \mid \mathbf{x})$. This takes on the form
\begin{equation}\label{eq:conditreg_conditional}
\begin{aligned}
p_{i}(y \mid \mathbf{x}) = \{&1-\alpha_{i}(\mathbf{x},\mathbf{x}_{i})+\\ &\alpha_{i}(\mathbf{x},\mathbf{x}_{i})\, d_{a}\left(q_{i-1},r_{i-1}\right)\}\, p_{i-1}(y\mid \mathbf{x})
\end{aligned}
\end{equation}
where $q_{i-1} = P_{i-1}(y \mid \mathbf{x}), \quad r_{i-1} = P_{i-1}(y_i \mid \mathbf{x}_i)$. 
The exact form of the function $\alpha_{i+1}(\mathbf{x},\mathbf{x}_{i+1})$, provided in Appendix \ref{sec:supp_meth}, can be derived from the multivariate copula update \citep{Fong2021}. Intuitively, $\alpha_{i}(\mathbf{x},\mathbf{x}_{i})$ weights the copula kernel based on the distance between the covariate of interest $\mathbf{x}$ and the updating datum $\mathbf{x}_i$. 

\subsubsection{Practical Details}
We now review the practical details discussed in \citet{Fong2021}. To estimate $p_n(y)$ before predictive resampling, we need to begin with $p_0(y)$, which acts as our prior guess of the true density. A choice that works well in practice is to set $p_0(y) = \text{Lomax}(a,1)$, which matches the \gls*{dpmm}. Here $a$ is the bandwidth parameter, which we can set by maximizing the prequential log-likelihood \citep{dawid1984}, $\sum_{i=1}^n \log p_{i-1}(y_i)$. 
Fitting the copula method has a computational complexity of $\mathcal{O}(n^2)$, as we must first compute the overhead terms $P_{i-1}(y_{i})$ for $i = 1,\ldots, n$. Given these terms, computing $p_n(y)$ at any value of interest is then $\mathcal{O}(n)$.

A key property of the copula methods is the convenience of predictive resampling. The copula update for $p_{n+1}$ only depends on $Y_{n+1}$ through $P_n(Y_{n+1})$, and predictive resampling involves drawing $Y_{n+1} \sim P_n$. As a result, we only need to draw $P_i(Y_{i+1}) \iid \mathcal{U}[0,1]$ for $i = n,\ldots,N-1$, and compute the copula update (\ref{eq:positive_update}) appropriately. Drawing a sample of $p_N(y)$ at some test point is thus $\mathcal{O}(N-n)$. In the regression case, we can draw $\mathbf{X}_{n+1:N}$ from the Bayesian bootstrap, and $P_i(Y_{i+1} \mid \mathbf{X}_{i+1})$ can be similarly drawn from the uniform distribution as in the no-covariate case.

\section{PREDICTIVE RESAMPLING UNDER RIGHT-CENSORING}
{The above description assumed that we observed each of the survival times exactly. We are now ready to extend the predictive resampling framework to right-censored data.} Suppose we have observed the dataset $\mathcal{D}_n:= \{y_{1:k}, Y_{k+1:n} \geq c_{k+1:n}\}$, where for convenience we have ordered the data such that the first $k$ are observed and the remaining are right-censored.
Throughout the remainder of this work, we assume that the censoring mechanism is \textit{non-informative} - that is, we treat $c_{k+1:n}$ as constants in conditional probability statements. See \citet{berliner1988} and Appendix \ref{sec:supp_theory} for more details on the relevant assumptions in the Bayesian and predictive cases.
Once again, the Bayesian requires $y_{1:N}$ (for $N\to \infty$) to compute any statistic of interest, and so it is natural that the Bayesian elicits the predictive density
\begin{equation}\label{eq:censored_future_joint}
p(y_{k+1:N} \mid \mathcal{D}_n)
\end{equation}
on $Y_{k+1:N}$ which is missing. In contrast to the uncensored case however, $Y_{k+1:n}$ is partially observed. The key is to factorize (\ref{eq:censored_future_joint}) into
\begin{equation}\label{eq:factorized_joint}
    p(y_{n+1:N} \mid y_{1:n}) \, p(y_{k+1:n} \mid \mathcal{D}_n),
\end{equation}
and so predictive resampling consists of the following:
\begin{enumerate}
    \item Impute $Y_{k+1:n} \sim p(y_{k+1:n} \mid \mathcal{D}_n)$.
    \item Predictive resample $Y_{n+1:N} \sim p(y_{n+1:N} \mid y_{1:n})$ as before.
    \item Compute a statistic of interest $\theta(Y_{1:N})$.
\end{enumerate}
The distribution of $\theta(Y_{1:N})$ is then approximately our martingale posterior distribution { $\pi_\infty(\theta \mid \mathcal{D}_n)$, where the subscript is used to distinguish from the regular Bayesian posterior.} We note that the exact martingale posterior distribution would involve computing the functional of the limiting empirical distribution; see Appendix \ref{sec:supp_theory} for more details. We also highlight the connection to the multiple imputation framework of \citet{Rubin1996}, where the full posterior $\pi(\theta \mid y_{1:n})$ is replaced by $p(y_{n+1:N} \mid y_{1:n})$ and the imputing predictive is given by $p(y_{k+1:n} \mid y_{1:k}, Y_{k+1:n} \geq c_{k+1:n}).$

\subsection{Doob's Consistency Theorem for Right-censored Observations}

As discussed above for the \gls*{iid} setting with fully observed data, it follows from Doob's consistency theorem \citep{Doob1949} that predictive resampling with the { parametric} posterior predictive distribution is equivalent to posterior sampling in the limit of $N\to \infty$. We now extend this result to the case where some of the observations are right-censored, as is typical for survival data.

Assume that for all $N$, the \glspl*{rv} $[\Theta,Y_1,\ldots,Y_N]$ have joint density
$$
p(\theta,y_{1:N}) = \pi(\theta)\, \prod_{i=1}^N f_\theta(y_i).
$$
Denote the posterior mean as $\bar{\theta}_N = E\left[\Theta \mid Y_{1:N} \right]$ (for $\Theta$ in a linear space), and $f^c_\theta(y) = \mathds{1}\lbrace y \geq c \rbrace f_\theta(y) / \bar{F}_\theta(c)$ 
to be the density of a data point right-censored at $c$, where $\bar{F}_\theta$ is the survival function of $f_\theta$.

We draw $Y_{k+1:n} \sim p(y_{k+1:n} \mid \mathcal{D}_n)$ where
$$
p(y_{k+1:n} \mid \mathcal{D}_n) = \int  \prod_{i=k+1}^n f^{c_i}_\theta(y_i) \, \pi(\theta \mid \mathcal{D}_n) \, d\theta,
$$
and
$\pi(\theta \mid \mathcal{D}_n) \propto \pi(\theta)\, \prod_{i=1}^k f_\theta(y_i) \prod_{i=k+1}^n \bar{F}_\theta(c_i),$ 
which follows from the non-informative censoring.

We then draw $Y_{n+1:N} \sim p(Y_{n+1:N} \mid y_{1:n})$  where
$$
p(y_{n+1:N} \mid y_{1:n}) = \int \prod_{i=n+1}^N f_\theta(y_i) \, \pi(\theta \mid y_{1:n}) \, d \theta,
$$
and compute $\bar{\theta}_N$ from $Y_{1:N}$.
The following result establishes the equivalence of predictive resampling and standard Bayesian inference as $N \to \infty$.

\begin{theorem} \label{thm:doob}
Assume $E[|\Theta| \mid \mathcal{D}_n] < \infty$. Under regularity conditions on $f_\theta$, we have that
\begin{equation}\label{eq:doob_surv}
\lim_{N\to \infty} \bar{\theta}_N =  \Theta \quad \textnormal{a.s. } P^\infty(\cdot \mid \mathcal{D}_n).
\end{equation}
where $P^\infty$ is over $\Theta$ and $Y_{k+1:\infty}$.
\end{theorem}
\vspace{-0.45cm}

\begin{proof}
See Appendix \ref{sec:supp_theory}.
\end{proof}

\vspace{-0.2cm}
Similarly to \citet{Fong2021}, the above theorem directly links Bayesian uncertainty in the parameter, represented by $\Theta \sim \pi(\theta \mid \mathcal{D}_n)$, to the statistical uncertainty 
in the \textit{partially} observed $Y_{k+1:n}$ and unobserved $Y_{n+1:\infty}$.
This can be seen by considering the following two distinct methods of sampling $\Theta$ from the posterior. The first is the standard Bayesian approach to draw $\Theta \sim \pi(\theta \mid \mathcal{D}_n)$ directly. The second, predictive resampling, begins by first imputing the partially observed data points $Y_{k+1:n}$ from the joint density $p(y_{k+1:n} \mid \mathcal{D}_n)$ followed by the completely unseen observables $Y_{n+1:\infty}$ from the sequence of predictive densities
$$
    Y_{n+1}  \sim p(\cdot \mid y_{1:n}),~ 
    Y_{n+2}  \sim p(\cdot \mid y_{1:n+1}),~
    \dots,
$$
until we have the complete information $Y_{1:\infty}$. Given $Y_{1:\infty}$, we can then compute the limiting estimate ${\bar{\theta}_\infty = \lim_{N \to \infty} \bar{\theta}_N}$, which is the posterior mean, on the entire dataset. By the above theorem, this returns $\bar{\theta}_\infty \sim \pi(\theta \mid \mathcal{D}_n)$. 

{ We emphasize that the purpose of Theorem 1 as outlined above is to provide a conceptual illustration that, in the Bayesian \textit{parametric} case, the uncertainty in a point estimator $\bar{\theta}_N$ computed from imputed observations is equivalent to uncertainty in the Bayesian random parameter $\Theta$. The choice of the posterior mean $\bar{\theta}_N$ as the estimator is one of mathematical convenience, allowing us to directly leverage the result of \citet{Doob1949}; it may not be of practical use when the posterior mean is not analytically available. In the more general martingale posterior case, the c.i.d. property  guarantees the existence of the limiting empirical distribution, $F_\infty$, under our imputation and predictive resampling scheme, again relying on martingales in an analogous way to Doob's theorem. We can then compute the functional of interest, $\theta(F_\infty)$, to obtain a posterior sample. To show this in the c.i.d. case, we condition on $Y_{k+1:n}$ and utilize the properties of the original c.i.d. sequence in a similar way to Theorem 1. Further details can be found in Appendix \ref{sec:supp_theory}.}

\section{COPULA UPDATES UNDER RIGHT-CENSORING}

The copula updates introduced in Section~\ref{sec:bivariate_copula} assumed that observations were fully known. We now extend these methods to the right-censored case. For the purposes of exposition, we will continue to treat the first $k$ data points $y_{1:k}$ as uncensored,  with the remaining $y_{k+1:n}$ as right-censored at $c_{k+1:n}$. In practice however, a random ordering is usually preferred, and we highlight that the copula methods are not exchangeable. See {Appendix \ref{sec:supp_meth}} for further discussion on ordering. 

If the aim is to predict survival outcomes for a new individual given right-censored observations, the quantity of interest is the predictive density
\begin{equation}\label{eq:censored_pred}
p(y_{n+1} \mid \mathcal{D}_n).
\end{equation}
This can be written as
\begin{equation}\label{eq:censored_pred_integral}
\int p(y_{n+1} \mid y_{1:n}) \,p(y_{k+1:n} \mid \mathcal{D}_n) \, dy_{k+1:n},
\end{equation}
which can be computed via Monte Carlo, where $p(y_{n+1} \mid y_{1:n})$ is available through (\ref{eq:positive_update}). To obtain both the martingale posterior and predictive density, we will now develop a method to simulate from 
\begin{equation}\label{eq:censored_joint}
    p(y_{k+1:n} \mid \mathcal{D}_n).
\end{equation}

\subsection{Importance Sampling}
To simulate from (\ref{eq:censored_joint}) given a prescribed sequence $\{p_{i-1}(y_i)\}_{i = 1:n}$, we draw inspiration from \citet{Kong1994}, which considered this problem for {fully} missing data by sequential imputation followed by importance reweighting. We now introduce the methodology in the particular case of right-censored data.

For the first $k$ data points, the update (\ref{eq:positive_update}) can be used recursively to obtain $p_k(y_{k+1})$. When we reach the first censored datum $Y_{k+1} \geq c_{k+1}$, we cannot directly update the predictive density as it requires the value of $Y_{k+1}$. An intuitive, but incorrect, solution is as follows: impute $Y_{k+1} \sim p(y_{k+1} \mid Y_{k+1} \geq c_{k+1}, y_{1:k})$, then treat the sampled $Y_{k+1}$ as an observed value to update to $p_{k+1}$ via (\ref{eq:positive_update}). Then, draw $Y_{k+2}$ from $p_{k+1}(y_{k+2}) = p(y_{k+2} \mid y_{1:k}, Y_{k+1},Y_{k+2} \geq c_{k+2})$ and continue on in a sequential manner until we have $Y_{k+1:n}$. 

However, this sample is not drawn from (\ref{eq:censored_joint}). In short, this is because we have not used the future censored information $\{Y_{j} \geq c_{j}\}_{j > i}$ when imputing $Y_i$, for $i = k+1,\ldots,n$.
To correct for this, we can use \gls*{IS}, treating $Y_{k+1:n}$ as a proposal sample. 
Assuming non-informative censoring, the importance weights can be derived through the factorization of (\ref{eq:censored_joint}) into 
\begin{equation}
\begin{aligned}
 &\frac{\prod_{i = k+1}^n {p(y_i, Y_i \geq c_i \mid y_{1:i-1})}}{p(Y_{k+1:n} \geq c_{k+1:n}\mid y_{1:k})}\\
&\propto \underbrace{\prod_{i=k+1}^n p(y_i \mid Y_i \geq c_i, y_{1:i-1})}_{\text{Proposal}} \underbrace{\prod_{i = k+1}^n P(Y_i \geq c_i \mid y_{1:i-1})}_{\text{Unnormalized \gls*{IS} weights}}.
\end{aligned}
\end{equation}
{In the above, we have used the notation 
\begin{align*}
    p(x, Y \geq c) :=& \, p(x \mid Y \geq c) \, P(Y \geq c)  \\
    =& P(Y \geq c \mid x) \, p(x),
\end{align*} 
to represent the mixed joint density of the observed values and censored events, where $x$ is a continuous \gls*{rv} and $\mathds{1}(Y \geq c)$ can be considered as a discrete \gls*{rv}. }
The proposal is the joint density of $Y_{k+1:n}$ drawn from the scheme above, and the \gls*{IS} weight can be computed as we sequentially impute, since it depends only on $p_i$.

For the copula method specifically, the proposal is efficient to simulate from sequentially in a rejection-free manner. Writing $p_{i-1}^{c_i}(y) = p(y \mid Y_i \geq c_i,y_{1:i-1})$, we note that 
 $p^{c_i}_{i-1}(y)  \propto \mathds{1}(y\geq c_i) \, p_{i-1}(y)$
   from non-informative censoring.
Working in the space of \gls*{cdf}s, simulating $Y_{i}\sim p_{i-1}^{c_i}$ is equivalent to drawing
\begin{equation*}
    \begin{aligned}
U_{i} &\sim \mathcal{U}[P_{i-1}(c_{i}),1 ], \quad Y_{i} = P_{i-1}^{-1}(U_{i}).
\end{aligned}
\end{equation*}
However, we note that the update (\ref{eq:positive_update}) depends on $Y_{i}$ only through $U_{i} = P_{i-1}(Y_{i})$, and so we can utilize $U_{i}$ directly without computing $P_{i-1}^{-1}$. The term $P_{i-1}(c_{i})$ is used both in the proposal and the \gls*{IS} weight, and can be computed exactly as a tractable update also exists for the \gls*{cdf} sequence.

Given $B$ samples from the proposal  $\{Y_{k+1:n}^{(j)}\}_{j = 1:B}$ and self-normalized \gls*{IS} weights given by
$$
{w}^{(j)} = \prod_{i = k+1}^n \left[ 1-P^{(j)}_{i-1}(c_i)\right], \quad \tilde{w}^{(j)} = w^{(j)} / \sum_{j=1}^B w^{(j)},
$$
we can then approximate (\ref{eq:censored_pred}) with
$$
 \hat{p}(y_{n+1} \mid \mathcal{D}_n) = \sum_{j=1}^B \tilde{w}^{(j)} \, p^{(j)}_n(y_{n+1}),
$$
where $p^{(j)}_n$ is the random predictive density computed from $\{y_{1:k}, Y_{k+1:n}^{(j)}\}$ through (\ref{eq:positive_update}). Similarly, we can approximate the martingale posterior through
$$
    \hat{\pi}_
    {\infty}(\theta \mid \mathcal{D}_n) =  \sum_{j=1}^B \tilde{w}^{(j)} \, \delta_{\theta_N^{(j)}}
$$
where $\theta_N^{(j)} = \theta(y_{1:k}, Y_{k+1:N}^{(j)})$ and the unobserved $Y^{(j)}_{n+1:N} \sim p(y_{n+1:N} \mid y_{1:k}, Y_{k+1:n}^{(j)})$ are simulated through regular predictive resampling after imputing $Y_{k+1:n}^{(j)}$.

\subsection{Sequential Monte Carlo}
If the number of missing data points $n-k$ is large, \gls*{IS} may fail due to the dimensionality of the proposal. To mitigate the exponential variance increase of vanilla \gls*{IS}, we can use \gls*{smc}, noting that the importance weights have a straightforward online update. This induces additional resampling steps of $\{w^{(1:B)}, Y_{k+1:i}^{(1:B)}\}$ at time points $i$ when the \gls*{ess} is too low, e.g. less than $50\%$ of the original number of particles. In practice, we find \gls*{smc} to drastically improve the performance of our method for larger values of $n-k$ for a minor increase in computation. See \citet{Doucet2009} for more details.

\vspace{-0.2cm}

\subsection{Algorithm and Practical Details}
In practice, we find that randomizing the order of data greatly increases the \gls*{ess} in comparison to ordering the uncensored data before the censored data. However, the IS weights in this case have a slightly different form to take into account observed data points between censored data points. This is shown in Algorithm \ref{alg:joint_sampling}, using the notation $\delta_i = 1$ to indicate that $y_i$ is observed and $\delta_i =0$ to indicate that $Y_i$ is right-censored at $c_i$. See Appendix \ref{sec:supp_meth} for the derivation and { more details on the impact of ordering on the \gls*{ess}}.

 To select the bandwidth $a$, we can maximize the joint likelihood of the observations, $p(\mathcal{D}_n)$, which can be computed with SMC (Appendix \ref{sec:supp_meth}).  As we are still required to compute $\{P_{i-1}(\delta_i y_i + (1-\delta_i)c_i)\}_{i =1:n}$ for each particle, the total complexity of Algorithm \ref{alg:joint_sampling} is $\mathcal{O}(Bn^2)$, followed by $\mathcal{O}(Bn)$ for each prediction. { Details on the selection of the number of forward samples $N$ to sufficiently approximate the infinite population are given in Appendix \ref{sec:supp_meth}}.
\begin{algorithm}[!h]\label{alg:joint_sampling}
\DontPrintSemicolon
  \SetAlgoLined
{Initialize $p_0$ and $w^{(1:B)} = 1$}\;
\For{$i \gets 1$ \textnormal{\textbf{to}} $n$}{
 \For{$j \gets 1$ \textnormal{\textbf{to}} $B$}{
\If{$\delta_i = 1$}{
Update $w^{(j)} = w^{(j)} \times p^{(j)}_{i-1}(y_i)$ \;
 Update  $p^{(j)}_{i} \mapsfrom \left\{p^{(j)}_{i-1}, y_i\right\}$ from (\ref{eq:positive_update})\; 
}
\If{$\delta_i = 0$}{
Sample $Y^{(j)}_{i}  \sim {p}^{c_i}_{i-1}$\;
  Update $w^{(j)} = w^{(j)} \times \left[1 - P^{(j)}_{i-1}(c_i) \right]$ \;
   Update  $p^{(j)}_{i} \mapsfrom \left\{p^{(j)}_{i-1}, Y^{(j)}_{i}\right\}$ from (\ref{eq:positive_update})\;
}
}
\If{$ESS(w^{(1:B)}) < 0.5B$}{
Resample $\{1/B,\bar{p}_n^{(1:B)}\} \mapsfrom \{w^{(1:B)},p_n^{(1:B)}\}$ 
}
}
 Return $\{ w^{(1:B)},p_n^{(1:B)}\}$\;
\caption{Survival Copula Imputation}
\end{algorithm}

\vspace{-0.4cm}

\subsection{Covariates}
In the context of survival regression, we are interested in the effect of observed covariates $\mathbf{x}_{1:n}$ on survival outcomes. The density of interest is now
${p(y_{k+1:n} \mid \mathcal{D}_n, \,\mathbf{x}_{1:n})}
$. Given a tractable sequence of conditional densities, $p_i(y \mid \mathbf{x})$, the importance reweighting method above generalizes easily (Appendix \ref{sec:supp_meth}).
With (ignorable) missing covariates, our reweighting scheme can be combined with that of \citet{Kong1994}.
Note that the terms $\alpha_i(\mathbf{x},\mathbf{x}_{i})$ also depend on a hyperparameter $\rho_x$. We can fit $\{a,\rho_x\}$ jointly by maximizing the conditional prequential log-likelihood.

\section{EXAMPLES}

We illustrate our approach on a simulated data example and three real data examples, including two with covariates, comparing our approach to common Bayesian nonparametric survival analysis methods.

All copula examples are implemented in JAX \citep{Jax2018} and run on an Azure NC6 Virtual Machine with a one-half Tesla K80 GPU card (with compilation times $<5$s). The \gls*{dpmm} examples are run on a 2.4 GHz 8-Core Intel Core i9-9980HK using the \texttt{R} packages \texttt{dirichletprocess} \citep{ross2018} and \texttt{ddpanova} \citep{deIorio2004}. For all examples, we have $B = 2000$ \gls*{IS}/\gls*{mcmc} samples and set $N = 2000 + n$ for the number of future samples { which is sufficiently large for convergence (Appendix \ref{sec:supp_results})}. For the copula methods, we use a single random permutation of the data for each run { and fit the bandwidth automatically by maximizing the prequential log-likelihood.} The code and data used is available online\footnote{\href{https://github.com/edfong/survival_mp}{\url{https://github.com/edfong/survival_mp}}}. Further details, such as evaluation of the \gls*{ess} and standardization, are provided in Appendix \ref{sec:supp_results}.

\subsection{Simulated Data}\label{sec:sim}

{We begin by providing an empirical illustration of Theorem~\ref{thm:doob}, that is Doob's consistency theorem for right-censored observations, through a simulated data example with a Bayesian \textit{parametric} predictive.} We generated data
\begin{equation*}
Y_i \sim \text{Exp}(1), \quad C_i \sim \text{Exp}(2), \quad i = 1, \dots, n,
\end{equation*}
and right-censor if $y_i \geq c_i$, with $n = 50$.  Around 76\% of data points were right-censored. We consider fitting this data with an $\text{Exp}(1/\theta)$ sampling density under a conjugate inverse-gamma prior $\text{IG}(\theta \mid a_0, b_0)$. The posterior is then
$\pi(\theta \mid \mathcal{D}_n) = \text{IG}(a_n, b_n)$, where
$
a_n = a_0 + k, \quad b_n = b_0 + \sum_{i=1}^k y_i + \sum_{i = k+1}^n c_i,
$
and the posterior predictive is also analytically tractable as the $ \text{Lomax}(a_n,b_n)$ distribution (Appendix \ref{sec:supp_results}). 

For the inverse-gamma prior, we set $b_0=1$ and select $a_0 = 1.2$ by maximizing the marginal likelihood. We perform the imputation of the censored data points as in Algorithm~\ref{alg:joint_sampling}, noting that the importance weights are available via the Lomax \gls*{cdf}. This is then followed by regular predictive resampling (see Appendix \ref{sec:supp_meth}). Figure~\ref{fig:sim_param_smc_truth} illustrates the close agreement between the standard Bayesian posterior and the martingale posterior induced by the Lomax posterior predictive distribution, as expected from Theorem~\ref{thm:doob}. 

\begin{figure}[!ht]
\centering
 \includegraphics[width=0.45\textwidth]{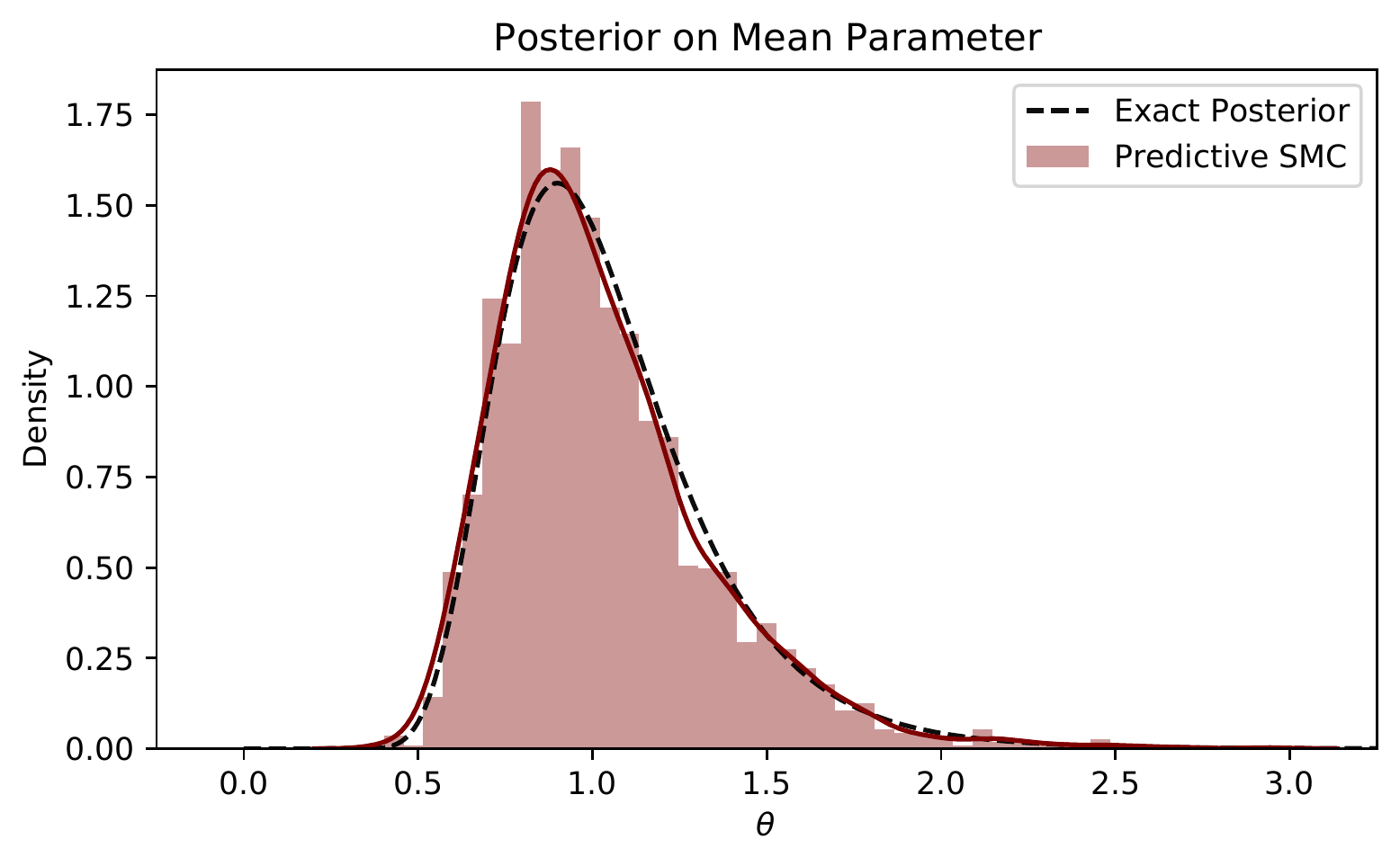}
 \vspace{-5mm}
  \caption{Exact Bayesian posterior vs. martingale posterior with parametric Lomax predictive generated via  Algorithm~\ref{alg:joint_sampling}.}\label{fig:sim_param_smc_truth} \vspace{-2mm}
\end{figure}

\subsection{Primary Biliary Cirrhosis}

We now shift attention to survival data from a randomized clinical trial on $n = 312$ patients with \gls*{pbc} \citep{dickson1989prognosis}, available in \texttt{R} through the \texttt{survival::pbc} dataset. A total of 158 patients received D-penicilammine, while the remaining 154 patients received a placebo. We compared our predictive resampling approach using the nonparametric exponential copula update with a \gls*{dpmm} using an exponential kernel. In particular, we focus on the predictive accuracy of these methods, evaluating performance on each of the trial arms separately. We first applied 10 random 50-50 train-test splits to the data, fit each model to the training set, then computed the mean log-likelihood of the resulting fit on the test set, which contains both censored and uncensored data points. The predictive accuracy of the two methods is almost identical (Table~\ref{tab:pbc}), and full details can be found in the Appendix \ref{sec:supp_results}.

In Figure \ref{fig:pbc2_posterior}a, we also plot the posterior mean and 95\% credible intervals of the survival function for the placebo arm, fitted to all 154 data points. We see that the results are similar, but the copula method has wider credible intervals. {See Appendix \ref{sec:supp_results} for posterior plots of the nonparametric density.} In Figure \ref{fig:pbc2_posterior}b, we plot posterior samples of the median, which is again similar. The copula method required 3.6s to optimize for hyperparameters and fit the data, and a further 0.9s
for predictive resampling on a $y$-grid of size $149$. In contrast, the \gls*{dpmm} took around 2 minutes.

\begin{center}
\begin{table}[h]
\begin{footnotesize}
\caption{Average Test Log-likelihood with Standard Errors (in Brackets) on the Two Arms of the \gls*{pbc} Dataset.\label{tab:pbc}}\vspace{1mm}
\begin{tabular}{ c c c }
 Dataset & Copula (exp) & \gls*{dpmm} (exp) \\
 \hline
  \gls*{pbc} (treatment) & -0.44 (0.02) & -0.43 (0.02) \\  
  \gls*{pbc} (placebo) & -0.39 (0.02) & -0.39 (0.03)    
\end{tabular}
\end{footnotesize}
\end{table}
\vspace{-10mm}
\end{center}

\begin{figure*}[!ht] 
\centering
 \includegraphics[width=0.85\textwidth]{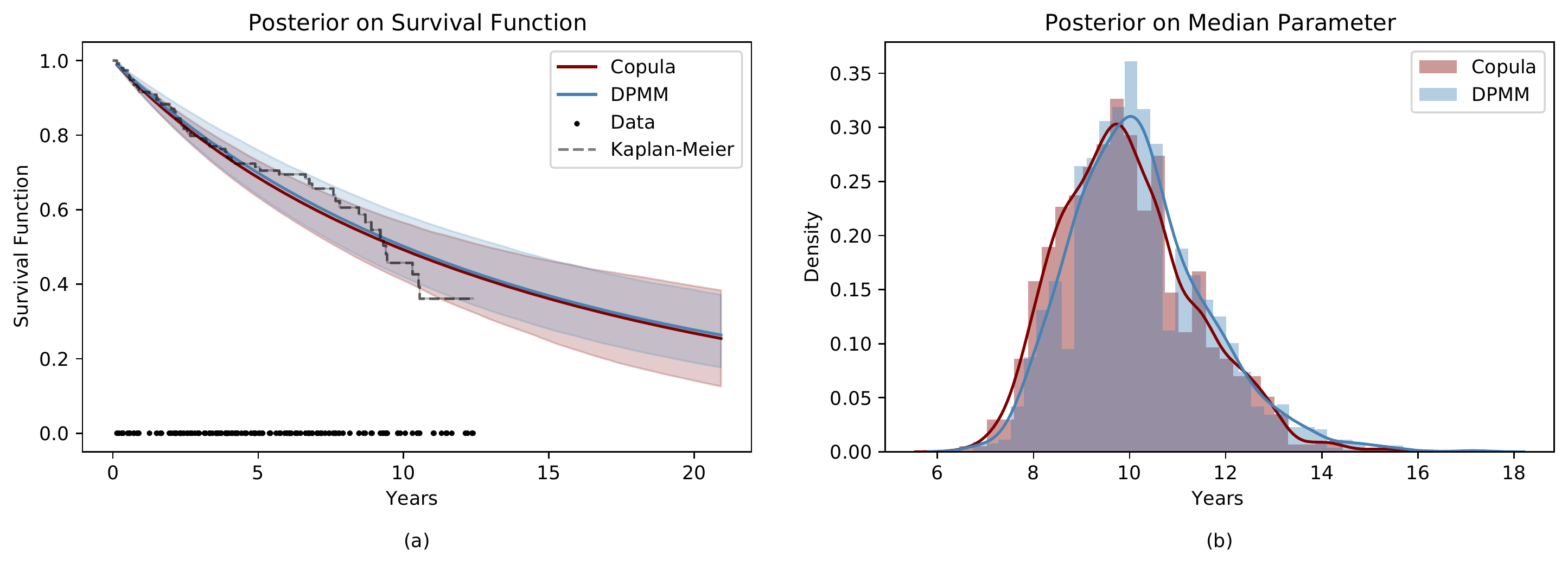}
 \vspace{-5mm}
  \caption{(a) Posterior mean and 95\% credible interval of the survival function and (b) samples of median for the \gls*{pbc} placebo arm.}\label{fig:pbc2_posterior}  \vspace{-2mm}
\end{figure*}

\subsection{Survival Regression}

Next, we illustrate our method for survival data in the presence of covariates. For this purpose, we analysed two datasets. First, we analysed data on $n = 205$ patients in Denmark with malignant melanoma, using tumour thickness as a covariate. Second, we analysed survival data of $n=863$ kidney transplant patients at The Ohio State University Transplant Center from 1982 to 1992, using patient age as a covariate. These datasets are available in \texttt{R} from the \texttt{MASS} \citep{mass} and \texttt{KMsurv} \citep{kmsurv} packages respectively.

For the baselines, we fit a \gls*{dpmm} with a log-normal kernel and an \gls*{aft} model with log-normal noise. For a fair comparison, we utilize a variant of the copula update,  substituting the Clayton copula in (\ref{eq:positive_update}) with the Gaussian copula, and setting $p_0(y) = \text{Log-normal}(y \mid 0, 1/(1-\rho))$. This corresponds to a copula update based on the log-normal \gls*{dpmm}; more details can be found in Appendix \ref{sec:supp_meth}.

Once again, we carried out 10 random 50-50 train-test splits and evaluated the predictive log-likelihood on the test set (Table \ref{tab:regression}). We see that the copula method performs the best for the melanoma dataset, but slightly worse than the other methods for the kidney dataset. 
Optimization, fitting and prediction for the copula method required around 3s and 14s for the melanoma and kidney dataset respectively, compared to 10s and 76s respectively for the DDP mixture, for each train-test split. It is also possible to predictively resample in the regression context (Appendix \ref{sec:supp_meth}). \vspace{-5mm}

\begin{center}
\begin{table}[h]
\begin{footnotesize}
\caption{Average Test Log-likelihood with Standard Errors (in Brackets) on the Melanoma and Kidney Datasets.\label{tab:regression}}\vspace{1mm}
\begin{tabular}{ c | c c c}
 Dataset & Copula & \gls*{dpmm} & AFT\\
 \hline
Melanoma &-0.22 (0.03)& -0.25 (0.02)& -0.23 (0.02)\\  
  Kidney & -0.11 (0.004) & -0.10 (0.004)  & -0.10 (0.003)  \\
\end{tabular}
\end{footnotesize}
\end{table}
\vspace{-7mm}
\end{center}

    For the melanoma dataset, we also visually evaluate the fit of the copula method on all $205$ data points (Figure \ref{fig:melanoma_surv_KM}). We follow the setup of \citet{RivaPalacio2021} and plot the predictive survival function for various tumour thicknesses $x$, comparing to the \gls*{km} estimator fit on windows centered around each $x$ value. The copula method matches reasonably closely with the stratified \gls*{km} estimator.{ Plots of the nonparametric median function can be found in Appendix \ref{sec:supp_results}.}

\begin{figure}[!h]
\centering
 \includegraphics[width=0.45\textwidth]{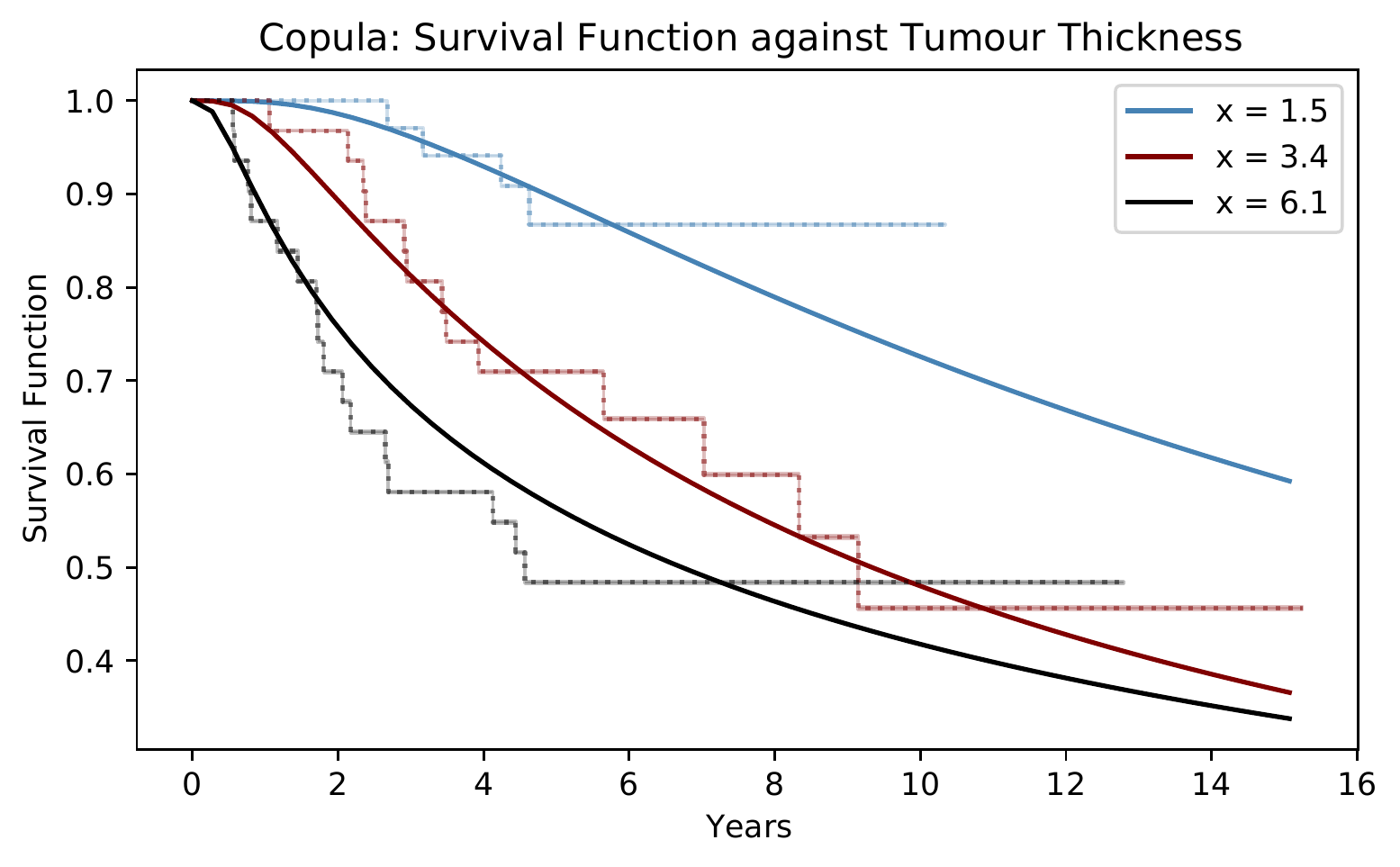}
 \vspace{-5mm}
  \caption{Survival function for copula method (\fullmidk) for $x =\{1.5,3.4,6.1\}$ with \gls*{km} (\dottedmidk) fit to windows $\{(1.255,1.75),(2.7,4.1),(4.1,8.1)\}$.}\label{fig:melanoma_surv_KM} 
\end{figure}








\section{DISCUSSION}

In this work, we have built on the martingale posterior distribution framework to the case in survival analysis where some of the observations are right-censored. 
We make use of one-step-ahead bivariate copula updates to perform inference, which admit a straightforward sequential importance sampling algorithm, thus avoiding the need for the likelihood/prior construction or \gls*{mcmc}. Our method is competitive with other Bayesian nonparametric survival models, both in terms of predictive accuracy and computation time. We note that a similar approach could be applied to other types of data with partially observed information. In future work, we hope to generalize by replacing imputation of the right-censored data with the imputation of $p(y_{\text{mis}} \mid y_{\text{obs}})$, where $y_{1:n} = \{y_{\text{obs}}, y_{\text{mis}}\}$.

There are a number of practical details regarding the implementation of our predictive resampling scheme that merit further investigation. Firstly, the computational complexity of the copula updates is $\mathcal{O}(n^2)$, which may be overly onerous for large datasets. Approximate methods such as subsampling may be one path towards reducing this computational cost. Secondly, we have found that predictive performance is, perhaps unsurprisingly, sensitive to the specification on the initial predictive density $p_0$. In our current scheme, the bandwidth parameter $a$ also controls the tails of $p_0$, which we set to be adaptive as a default value is difficult under right-censoring. Finally, we use grid-based optimisation to select the copula update hyperparameters. Although this approach was suitable for the examples studied here, it would not scale well to settings with a larger number of hyperparameters. Potential alternatives include stochastic gradient descent methods or a theoretically-justified plug-in selection procedure.

{We conclude by relating the martingale posterior framework to the \textit{operational subjective} approach to statistical inference \citep{lad1996}. The operational subjectivist specifies uncertain knowledge about quantities of interest through personal probability assertions, termed ``previsions'', relying on the fundamental theorem of prevision \citep{definetti1937, lad1990, lad1992} to ensure the coherence of a set of assertions. By assuming the exchangeability of observations, De Finetti's representation theorem provides a concise characterisation of the predictive via a common sampling density \citep{definetti1937}. In contrast, \citet{Fong2021} proposes to elicit the predictive directly via one-step-ahead copula updates. Under right-censoring, computational difficulty arises from the partial nature of the information associated with censored observations, which we resolve by using a Monte Carlo approximation that first imputes the missing information. While not pursued here, we note that it may be possible to bypass the Monte Carlo approximation by specifying a predictive that directly accounts for censored observations.}

\subsubsection*{Acknowledgements}

We thank the area chairs and the reviewers for their constructive feedback comments that have greatly improved this manuscript. We gratefully acknowledge invaluable guidance and support, especially in the initial stages of this work, from Chris Holmes. We also thank Stephen Walker for helpful discussions on predictive resampling for censored data, George Deligiannidis for assistance with the extension of Doob's consistency theorem, and Andrew Yiu for useful comments on the manuscript. E.F. was funded by The Alan Turing Institute Doctoral Studentship, under the EPSRC grant EP/N510129/1, and is currently employed at Novo Nordisk. B.L. was supported by the UK Engineering and Physical Sciences Research Council through the Bayes4Health programme (grant number EP/R018561/1) and gratefully acknowledges funding from Jesus College, Oxford.












\printbibliography


\clearpage
\appendix

\thispagestyle{empty}

\onecolumn \makesupplementtitle

\section{THEORY}\label{sec:supp_theory}

\subsection{Non-informative Censoring}\label{sec:supp_noninf_cens}
To illustrate the idea of non-informative censoring, consider the 
example of a single censored datum, $Y_1 \geq C_1$ where we observe $C_1 = c_1$. The usual random censoring assumption is
$$
Y_1 \sim F_\theta, \quad C_1 \sim G_\lambda.
$$
with $Y_1,~C_1$ independent. Under this assumption, the censoring mechanism is already non-informative for the \gls*{mle} of $\theta$, as the estimate of $\theta$ does not depend on $G_\lambda$. In the Bayesian case the additional assumption of prior independence,  $\pi(\theta, \lambda) = \pi(\theta)\, \pi(\lambda)$, is sufficient for the censoring mechanism to be non-informative. It is straightforward to show that the posterior in this case does not rely on $G_\lambda$, as it takes the form
\begin{equation*}
    \begin{aligned}
    \pi(\theta \mid Y_1 \geq C_1, c_1) &\propto {p(\theta, Y_1 \geq C_1, c_1)} \\ 
    &= \int p(\theta,\lambda, Y_1 \geq C_1, c_1) \, d\lambda\\ 
    &= \pi(\theta) \, \bar{F}_\theta(c_1) \, \int \pi(\lambda) \, g_\lambda(c_1) \, d\lambda\\
    &\propto \pi(\theta)\, \bar{F}_\theta(c_1).
    \end{aligned}
\end{equation*}
In the above, we have used the notation 
$p(x, Y \geq c) =p(x \mid Y \geq c) \, P(Y \geq c)= P(Y \geq c \mid x) \, p(x)$ to represent the mixed joint probability density function of the observed values and censored events, where $x$ is a continuous \gls*{rv} and $\mathds{1}(Y \geq c)$ can be considered as a discrete \gls*{rv}. We will continue to do so for the remainder of the Appendix in other contexts.

In predictive resampling, $p(y_{2:\infty}\mid y_1)$  by definition does not depend on the censoring. The censoring mechanism can thus only affect the imputing density $p(y_1 \mid Y_1 \geq C_1, c_1)$. Defining $p(y_1 \mid Y_1 \geq c_1) \propto \mathds{1}(y_1 \geq c_1)\, p(y_1)$, the Bayesian assumptions for non-informative censoring implies
$$
p(y_1 \mid Y_1 \geq C_1, c_1) = p(y_1 \mid Y_1 \geq c_1).
$$
We can see this through the following:
\begin{equation*}
    \begin{aligned}
    p(y_1 \mid Y_1 \geq C_1, c_1) &\propto {p(y_1, Y_1 \geq C_1,c_1)} \\
    &= P(Y_1 \geq C_1 \mid y_1, c_1)\, p(y_1,c_1) \\
    &= \mathds{1}(y_1 \geq c_1) \, p(y_1,c_1)
    \\&= \mathds{1}(y_1 \geq c_1) \, p(y_1)\, p(c_1)
    \\ &\propto \mathds{1}(y_1 \geq c_1) \, p(y_1).
    \end{aligned}
\end{equation*}
In the above,  $p(y_1,c_1) = \int f_\theta(y_1) \, g_\lambda(c_1) \, d\pi(\theta,\lambda)$. The key is that $p(y_1,c_1)$  factorizes into $p(y_1) \, p(c_1)$ under the assumptions of random censoring and independence in the prior.

In the absence of the likelihood and prior, a sufficient condition for $p(y_1 \mid Y_1 \geq C_1, c_1) = p(y_1 \mid Y_1 \geq c_1)$ is the factorization $p(y_1, c_1) = p(y_1) \, p(c_1)$, that is $Y_1$ and $C_1$ are a priori independent. More generally for the martingale posterior, a sufficient assumption for non-informative censoring is that under our predictive distribution, the vector $Y_{1:N}$ is independent of $C_{1:N}$ for all $N$. For the remainder of the Appendix, we continue to assume this and will treat the censoring times $c_i$ as constants.

\subsection{Doob's Consistency Theorem for Right-censored Observations \label{sec:supp_doob} }
In this section, we prove Doob's consistency theorem for the setting where some of the observations may be censored. As a reminder, we have $\mathcal{D}_n:= \{y_{1:k}, Y_{k+1:n} \geq c_{k+1:n}\}$. As $y_{1:k}$ are fully observed, we consider their values as fixed constants. For completeness, we include a repeat of the setup here.

Assume that for all $N$, the \glspl*{rv} $[\Theta,Y_1,\ldots,Y_N]$ have joint density
$$
p(\theta,y_{1:N}) = \pi(\theta)\, \prod_{i=1}^N f_\theta(y_i).
$$
We will make use of the standard Doob's consistency theorem (for uncensored observations) and so require the usual identifiability and measurability assumptions on the parametric sampling density $f_\theta$, which can be found in \citet{Doob1949} or \citet[Theorem~6.9, Proposition~6.10]{ghosal2017}. { Specifically, the identifiability condition is such that $F_\theta \neq F_{\theta'}$ whenever $\theta \neq \theta'$, where $F_\theta$ is the cumulative distribution function of $f_\theta$. This ensures that the parameter $\Theta$ can be recovered from the infinite sample.}
 
We assume that $\Theta$ lies in a linear space so the expectation is well-defined, and all regular conditional probability measures exist. We then write the posterior mean as $\bar{\theta}_N = E\left[\Theta \mid Y_{1:N} \right]$, and denote $f^c_\theta(y) = \mathds{1}\lbrace y \geq c \rbrace f_\theta(y) / \bar{F}_\theta(c)$ 
to be the density of a data point right-censored at $c$, where $\bar{F}_\theta$ is the survival function of $f_\theta$.

We draw $Y_{k+1:n} \sim p(y_{k+1:n} \mid \mathcal{D}_n)$ where
$$
p(y_{k+1:n} \mid \mathcal{D}_n) = \int  \prod_{i=k+1}^n f^{c_i}_\theta(y_i) \, \pi(\theta \mid \mathcal{D}_n) \, d\theta,
$$
and
$\pi(\theta \mid \mathcal{D}_n) \propto \pi(\theta)\, \prod_{i=1}^k f_\theta(y_i) \prod_{i=k+1}^n \bar{F}_\theta(c_i),$ 
which follows from the non-informative censoring assumption.

We then draw $Y_{n+1:N} \sim p(Y_{n+1:N} \mid y_{1:n})$  where
$$
p(y_{n+1:N} \mid y_{1:n}) = \int \prod_{i=n+1}^N f_\theta(y_i) \, \pi(\theta \mid y_{1:n}) \, d \theta,
$$
and compute $\bar{\theta}_N$ from $\{y_{1:k},Y_{k+1:N}\}$.
The following result establishes the equivalence of predictive resampling and standard Bayesian inference as $N \to \infty$.
\setcounter{theorem}{0}
\begin{theorem} 
Assume $E[|\Theta| \mid \mathcal{D}_n] < \infty$. Under regularity conditions on $f_\theta$, we have that
\begin{equation}\label{eq:supp_doob_surv}
\lim_{N\to \infty} \bar{\theta}_N =  \Theta \quad \textnormal{a.s. } P^\infty(\cdot \mid \mathcal{D}_n),
\end{equation}
where $P^\infty$ is over $\Theta$ and $Y_{k+1:\infty}$.
\end{theorem}
\begin{proof}
For each $y_{k+1:n} \in \mathbb{R}^{n-k}$ such that $E[|\Theta| \mid y_{1:n}] < \infty$, Doob's consistency theorem gives us
\begin{equation}\label{eq:supp_doob}
\lim_{N\to \infty} \bar{\theta}_N = \Theta \quad \textnormal{a.s. }P^\infty(\cdot \mid y_{1:n}).
\end{equation}
Note that the tower rule gives us
$$E[E[|\Theta| \mid y_{1:n}] \mid \mathcal{D}_n] = E[|\Theta| \mid\mathcal{D}_n] < \infty,$$
so $E\left[|\Theta| \mid y_{1:n}\right] < \infty$ for $P(\cdot \mid \mathcal{D}_n)$-almost all $y_{k+1:n}$. 
This implies that (\ref{eq:supp_doob}) holds for $P(\cdot \mid \mathcal{D}_n)$-almost all $y_{k+1:n}$. Finally, we have the following:
\begin{equation*}
    \begin{aligned}
     P^\infty\left(\lim_{N} \bar{\theta}_N =  \Theta \mid \mathcal{D}_n\right) = &\int  \, \mathds{1} \left\{\lim_N \bar{\theta}_N = \Theta \right\} \, dP^\infty(\Theta, y_{k+1:\infty} \mid \mathcal{D}_n) \\
     =&\int \underbrace{\int \, \mathds{1} \left\{\lim_N \bar{\theta}_N = \Theta \right\} \, dP^\infty(\Theta, y_{n+1:\infty} \mid y_{1:n})}_{=1 \text{ a.s. }P(y_{k+1:n} \mid \mathcal{D}_n)} \, dP(y_{k+1:n} \mid \mathcal{D}_n) \\= & ~1,
    \end{aligned}
\end{equation*}
 which is exactly statement (\ref{eq:supp_doob_surv}). 
\end{proof}
{
\subsection{Conditionally Identically Distributed Sequences for Right-censored Observations}\label{sec:supp_cid}
In this section, we first review the c.i.d. properties in the fully observed case, before discussing the implications of the c.i.d. property when there is right-censoring in the observations. In the fully observed case, assume that the sequence of \glspl*{rv} $[Y_{n+1},\ldots,Y_N]$ is c.i.d., where $P_i(y)$ is the usual predictive cumulative distribution function of $Y_{i+1}$ conditional on $Y_{1:i}$ for $i \geq n$. Following \citet{Berti2004,Fong2021}, the sequence of predictive cumulative distribution functions is a martingale as it satisfies
$$
E[P_i(y) \mid y_{1:{i-1}}] = P_{i-1}(y)
$$
almost surely for each $y \in \mathbb{R}$, for $i \geq n$. We highlight that $P_i$ is the predictive distribution conditional on $y_{1:i}$, so the expectation is over $Y_{i}$. From the properties of the c.i.d. sequence \citep[Lemma 2.1, 2.4]{Berti2004}, we have that the predictive distribution converges weakly to a random probability distribution, $P_\infty$, almost surely, that is
$$
P_N(y) \to P_\infty(y)  \quad \textnormal{a.s. }P^\infty(\cdot \mid y_{1:n})
$$
for each $y \in \mathbb{R}$, where $P^\infty( \cdot \mid y_{1:n})$ is over $Y_{n+1:\infty}$. Furthermore, we have that $E[P_\infty(y) \mid y_{1:n}] = P_n(y)$ almost surely for each $y \in \mathbb{R}$, which is the unbiasedness coherence condition from \citet{Fong2021}. The  empirical distribution,
$$
F_N(y) = \frac{1}{N}\left\{\sum_{i=1}^n  \mathds{1}(y_i \leq y) + \sum_{i=n+1}^N \mathds{1}(Y_i \leq y) \right\},
$$
also satisfies the same property \citep[Theorem 2.2]{Berti2004}, that is
$$
F_N(y) \to F_\infty(y) \quad \textnormal{a.s. }P^\infty(\cdot \mid y_{1:n})
$$
for each $y \in \mathbb{R}$, and in fact $F_\infty = P_\infty$ almost surely. 

Returning to the right-censored case where $Y_{k+1:n}$ is in fact random and drawn from $p(y_{k+1:n} \mid \mathcal{D}_n)$, we note that the above convergence holds for each $y_{k+1:n}$. We can therefore write
\begin{equation*}
    \begin{aligned}
    P^\infty\left(\lim_{N} F_N(y) = F_\infty(y) \mid \mathcal{D}_n \right) = &\int  \, \mathds{1} \left\{\lim_{N} F_N(y) = F_\infty(y) \right\} \, dP^\infty(y_{k+1:\infty} \mid \mathcal{D}_n) \\
     =&\int \underbrace{\int \, \mathds{1} \left\{\lim_N  F_N(y) = F_\infty(y) \right\} \, dP^\infty(y_{n+1:\infty} \mid y_{1:n})}_{=1} \, dP(y_{k+1:n} \mid \mathcal{D}_n) \\= & ~1.
    \end{aligned}
\end{equation*}
A similar result can be shown for the limiting predictive distribution $P_\infty$. From the above, we see that the limiting empirical distribution exists under the imputation and predictive resampling scheme due to the martingale property of the c.i.d. sequence.  We can thus compute $\theta(F_\infty)$ to obtain a posterior sample from the martingale posterior, $\pi_\infty(\theta  \mid \mathcal{D}_n)$, where we use the subscript in $\pi_\infty$ to distinguish from the regular Bayesian posterior. Note here that  $F_\infty$ is unknown but we can obtain samples through predictive resampling, in contrast to the parametric case of Doob's theorem where $\bar{\theta}_\infty = \Theta$ is known.

Interestingly, the unbiasedness coherence condition of \citet{Fong2021} is also satisfied in the right-censored case. We can compute the posterior mean of $P_\infty$ as
\begin{equation*}
    \begin{aligned}
    E\left[P_\infty(y) \mid \mathcal{D}_n \right] &= E\left[ E\left[ P_\infty(y) \mid y_{1:n}\right] \mid \mathcal{D}_n \right]\\
    &= E\left[P_n(y) \mid \mathcal{D}_n\right]\\
    &= P(Y_{n+1} \leq y \mid \mathcal{D}_n),
    \end{aligned}
\end{equation*}
which is the cumulative distribution function of (\ref{eq:censored_pred}).
Note that the outer expectation in the first line is over $Y_{k+1:n} \sim p(y_{k+1:n} \mid \mathcal{D}_n)$, whereas the inner expectation is over $Y_{n+1:\infty} \sim p(y_{n+1:\infty} \mid y_{1:n})$. Once again, the posterior mean of $P_\infty$ is our best estimate of the distribution of $Y_{n+1}$ given $\mathcal{D}_n$, and our imputation and predictive resampling scheme has incurred no bias. To summarize, we inherit the nice coherency properties of \citet{Fong2021} as we have the c.i.d. sequence conditional on the imputed $Y_{k+1:n}$. 

A final point is that even with $Y_{k+1:n}$ marginalized out, the sequence $Y_{n+1:\infty}$ remains c.i.d. as the marginalized predictive distribution is just a mixture of the fully observed $P_i$. However, the predictive distribution when $Y_{k+1:n}$ is marginalized is not tractable, which is why we introduce the sequential  Monte Carlo scheme in our paper. Eliciting the marginalized predictive directly would be an interesting avenue of future work.

}
\section{METHODOLOGY \label{sec:supp_meth} }

\subsection{Predictive Resampling in the Uncensored Case}
Predictive resampling for the uncensored case, as described in \citet{Fong2021}, is given by Algorithm \ref{alg:predictive_resampling}. A slight intricacy is that the limiting empirical $F_\infty$ may be continuous but $F_N$ is always discrete. We opt instead to use the final random predictive $P_N$ as an estimate of the limiting empirical $F_\infty$, as it can be continuous and in fact converges to $F_\infty$ in the limit of $N \to \infty$ \citep{Berti2004} as discussed in Appendix \ref{sec:supp_cid}. A martingale posterior sample can then be computed as $\theta_N = \theta(P_N)$.

\begin{figure}[ht]
  \centering
  \begin{minipage}{.55\linewidth}
\begin{algorithm}[H]\label{alg:predictive_resampling}
\DontPrintSemicolon
  \SetAlgoLined
  {Compute $p_n$ from $y_{1:n}$}\;
  \For{$j \gets 1$ \textnormal{\textbf{to}} $B$}{
  \For{$i \gets n+1$ \textnormal{\textbf{to}} $N$} {
  Sample $Y_{i}  \sim {P}_{i-1}$\;
  Update  $P_{i} \mapsfrom \left\{P_{i-1}, Y_{i}\right\} $\;
  }
 Evaluate   ${\theta}^{(j)}_N =\theta(P_N)$ \;}
 Return $\{\theta_N^{(1)},\ldots,\theta_N^{(B)} \}$\;
\caption{Predictive Resampling \citep{Fong2021} }
\end{algorithm}
  \end{minipage}
\end{figure}

\subsection{Copula Updates}
The copula update for the first step of the \gls*{dpmm} has been derived previously in \citet{Hahn2018} and \citet{Fong2021}. We will restate the key details that are not included in the main paper here.

From the main paper, the copula update for the densities is 
\begin{equation}\label{eq:supp_positive_update}
\begin{aligned}
    p_{i+1}(y) &= \left[ 1-\alpha_{i+1} + \alpha_{i+1}d_{a}\left\{P_i(y),P_i(y_{i+1}) \right\}\right] p_i(y).
\end{aligned}
\end{equation}
If the kernel of the \gls*{dpmm} has density $f_\theta(y)$, and the base measure of the DP has centering measure $\pi(\theta)$, then the bivariate copula density is
\begin{equation}\label{eq:supp_bivariate_copula_density}
d_a(u,v) = \frac{\int \, f_\theta\{P_0^{-1}(u) \} \, f_\theta\{P_0^{-1}(v)\} \, \pi(\theta) \, d \theta}{p_0\{P_0^{-1}(u)\} \, p_0\{P_0^{-1}(v)\} },
\end{equation}
where $p_0(y) = \int f_\theta(y) \, \pi(\theta) \, d\theta$ and $a$ is a hyperparameter that depends on the specification of the likelihood and prior. We then have $P_0(y) = \int^y p(y')\, dy'$ and $P_0^{-1}$ is the inverse \gls*{cdf}. 

Note that the update (\ref{eq:supp_positive_update}) requires the \gls*{cdf} $P_i(y)$. Fortunately this update is typically tractable, and involves integrating (\ref{eq:supp_positive_update}):
\begin{equation*}
    \begin{aligned}
    P_{i+1}(y) &= (1-\alpha_{i+1}) \, P_i(y) + \alpha_{i+1} \int^y d_a\{P_i(y'), P_i(y_{i+1})\} \, p_i(y') \, dy'\\
    &= (1-\alpha_{i+1}) \, P_i(y) + \alpha_{i+1} \int_0^{P_i(y)} d_a\{u', P_i(y_{i+1})\} \, du'\\
    &=  (1-\alpha_{i+1}) \, P_i(y) + \alpha_{i+1} \, I_a\{ P_i(y),P_i(y_{i+1})\}.
    \end{aligned}
\end{equation*}

The second line follows from the change of variables $u' = P_i(y')$, and we have that 
\begin{equation}\label{eq:supp_partial_copula}
I_a(u,v) = \int_0^u d_a(u',v) \, du'.
\end{equation}
If $\pi$ is conjugate to $f_\theta$, then the forms of $I_a$ and $d_a$ are typically tractable.


\subsection{Exponential Copula Update}\label{sec:supp_exp_copula}
In this section, we derive the copula density corresponding to the \gls*{dpmm} with the exponential kernel and gamma centering measure, that is
 \begin{equation}\label{eq:supp_exp_likelihood}
    f_\theta(y) = \theta \exp\left(-\theta y \right), \quad \pi(\theta) = \text{Gamma}(\theta \mid a,b)
 \end{equation} 
 for $y \geq 0$. We can derive the copula by considering
\begin{equation}
\begin{aligned}
\int f_\theta(y)\, f_\theta(y_1)\, d\pi(\theta)
&= \frac{b^{a}}{\Gamma(a)} \int_0^\infty \theta^{a+1} \exp\left\{-(b+ y+ y_1) \theta \right\} d\theta \\ 
&= \frac{b^{a}}{\Gamma(a)\left(b+y+y_1 \right)^{a + 2}} \int_0^\infty x^{a+1} \exp\left(-x \right) dx \\
&= \frac{a(a+1)}{b^2}\left(1 + \frac{y+y_1}{b} \right)^{-(a + 2)}
\end{aligned}
\end{equation}
where in the second line we have used the substitution $x = \theta(b+y+y_1)$, and the third line uses ${\Gamma(a+2) = a\,(a+1)\,  \Gamma(a)}$. We also have
\begin{equation}\label{eq:supp_lomax_pdf}
p_0(y)= \int f_\theta(y)\, d\pi(\theta) =\frac{a}{b}\left(1 + \frac{y}{b} \right)^{-(a+1)}
\end{equation}
which is the $\text{Lomax}(a,b)$ density. The copula density then takes the form
\begin{equation*}
\begin{aligned}
d_{a,b}(y,y_1) &= \frac{a + 1}{a} \frac{ \left(1 + \frac{y}{b} \right)^{a+1}{\left(1 + \frac{y_1}{b} \right)^{a+1}}}{\left(1 + \frac{y+y_1}{b} \right)^{a+2}}\cdot
\end{aligned}
\end{equation*}
We would like the density as a function of $(u,v)$. To this end, note that the marginal \gls*{cdf} and inverse \gls*{cdf} are
\begin{equation}\label{eq:supp_lomax_cdf}
P_0(y) = 1 - \left(1+ \frac{y}{b}\right)^{-a},\quad P_0^{-1}(u) = b\left\{(1-u)^{-\frac{1}{a}} - 1 \right\}.
\end{equation}

Finally, this gives us the copula density
\begin{equation}
d_a(u,v) = \frac{a + 1}{a} \frac{(1-u)^{-\frac{a + 1}{a}}(1-v)^{-\frac{a + 1}{a}}}{\left\{(1-u)^{-\frac{1}{a}} + (1-v)^{-\frac{1}{a}} -1	\right\}^{a + 2}}
\end{equation}
where $u = P(y)$ and $v = P(y_1)$. 

To derive $I_a(u,v)$, we compute the integral in (\ref{eq:supp_partial_copula}). Substituting $x=(1-u')^{-\frac{1}{a}}+(1-v)^{-\frac{1}{a}}-1$ gives us
$$dx = du' \times \frac{1}{a}(1-u')^{-\frac{a+1}{a}},  \quad c_1 = (1-v)^{-\frac{1}{a}}, \quad c_2 =(1-u)^{-\frac{1}{a}} + (1-v)^{-\frac{1}{a}} -1. $$ 
Plugging this into  (\ref{eq:supp_partial_copula}) gives us
\begin{equation*}
    \begin{aligned}
    I_a(u,v) &= (a+1)(1-v)^{-\frac{a+1}{a}}\int_{c_1}^{c_2} x^{-(a+2)} \, dx\\
    &= 1- \frac{(1-v)^{-\frac{a+1}{a}}}{\left\{(1-u)^{-\frac{1}{a}} + (1-v)^{-\frac{1}{a}} 	-1\right\}^{a + 1}}.
    \end{aligned}
\end{equation*}
In practice, we set $b = 1$ and so $p_0(y) = \text{Lomax}(a,1)$.

\subsection{Log-normal Copula Update}
We derive the copula update for the \gls*{dpmm} with a log-normal kernel with a normal base measure, that is
$$
f_\theta(y) = \frac{\mathcal{N}(\log(y) \mid \theta, 1)}{y}, \quad \pi(\theta) = \mathcal{N}(\theta \mid 0,\tau^{-1})
$$
for $y\geq 0$. Working with $z = \exp(y)$, we get
\begin{equation*}
\begin{aligned}
\int f_\theta(y)\, f_\theta(y_1)\, d\pi(\theta)
&= \frac{\int \mathcal{N}(z \mid \theta,1) \, \mathcal{N}(z_1 \mid \theta,1) \, d\mathcal{N}(\theta \mid 0,1)}{\log(z) \, \log(z_1)}\\
\end{aligned}
\end{equation*}
Similarly, we have
$$p_0(y) =
\frac{\int \mathcal{N}(z \mid \theta,1)\,  d \mathcal{N}(\theta \mid 0,1)}{\log(z)}\cdot $$
Plugging the above into (\ref{eq:supp_bivariate_copula_density}) gives us the bivariate Gaussian copula density $c_\rho(u,v)$, as the $\log(z)$ terms cancel out. We write $c_\rho$ instead of $d_a$ to remain consistent with \citet{Hahn2018} and \citet{Fong2021}.

The Gaussian copula density is
\begin{equation}\label{eq:supp_gaussian_copula}
c_\rho(u,v) = \frac{\mathcal{N}_2\left\{ \Phi^{-1}(u),\Phi^{-1}(v)\mid 0,1,\rho\right\}	}{\mathcal{N}\{\Phi^{-1}(u) \mid 0,1\} \,\mathcal{N}\{\Phi^{-1}(v) \mid 0,1\}}
\end{equation}
where $\rho \in (0,1)$, $\Phi$ is the standard normal \gls*{cdf}, and $\mathcal{N}_2(0,1,\rho)$ is the bivariate normal density with mean 0, variance 1 and correlation $\rho$. We can similarly compute $H_\rho(u,v) = \int_0^u c_\rho(u',v) \,du'$, which is
\begin{equation}
H_\rho(u,v) = \Phi\left\{ \frac{\Phi^{-1}(u) - \rho \Phi^{-1}(v)}{\sqrt{1-\rho^2}}\right\}.
\end{equation}

Although this implies the same copula update as the one on $\mathbb{R}$ as introduced in \citet{Hahn2018}, the key difference is $p_0(y)$, which can be shown to be
$$
p_0(y) = \text{Log-normal}\left(y \mid 0, \frac{1}{1-\rho}\right).
$$

\subsection{Ordering for Copula Method}\label{sec:supp_ordering}
\citet{Kong1994} suggests ordering the data such that the observed data comes before the missing data, which is to ensure the proposal is close to the target for importance weight stability. In the right-censoring case however, we have found that this intuition does not extend. In practice, randomizing the order of data greatly increases the \gls*{ess} in comparison to ordering the uncensored data before the censored data. Although one can average the results over different permutations, we find that a single permutation works well in practice. In the random order case, the \gls*{IS} weights have a slightly different form to take into account the observed data points between censored data points - this is provided in Algorithm 1 and derived in Section \ref{sec:supp_IS_deriv}.

We postulate that ordering the data is undesirable due to the nature of right-censoring: as the uncensored $y_{1:k}$ will tend to take on smaller values, a density estimate constructed from $y_{1:k}$ will not be sufficiently right-skewed compared to the target distribution, which has support on the larger values $y_{k+1:n}$ that have been right-censored. This results in the proposal being too light-tailed with respect to the target distribution, leading to {IS} weights with high variance. We recommend randomizing the order as it results in a heavier-tailed proposal, and this works much better in practice. 

We demonstrate this in the parametric example of Section \ref{sec:sim}, where the joint density on $Y_{1:N}$ is exchangeable, so the ordering only affects importance weight stability. We compare the  \gls*{ess} of the \gls*{IS} weights for random data ordering versus the ordering $\{y_{1:k}, Y_{k+1:n}\geq c_{k+1:n}\}$, which is computed as 
\begin{equation}\label{eq:supp_ESS}
    \text{ESS}(w^{(1:B)}) = 1/\sum_{j=1}^B{\{w^{(j)}\}^2}\cdot 
\end{equation}
We carry out Algorithm \ref{alg:joint_sampling} followed by predictive resampling \textit{without} the \gls*{smc} resampling steps for the two orderings. As we see in Figures \ref{fig:sim_param_IS}a and \ref{fig:sim_param_IS_ord}a, the random ordering case is quite close to the truth even without \gls*{smc} resampling, but the uncensored/censored ordering case is a poor approximation. As expected, the \gls*{ess} for the random and uncensored/censored cases are 967 and 8  respectively for $B = 2000$. To visualize the cause, we see in Figure \ref{fig:sim_param_IS_ord}b that the proposal has poor support over the true posterior of $\theta$ as it is peaked and not sufficiently right-skewed. On the other hand, the random ordering case proposal in Figure \ref{fig:sim_param_IS}b has a  heavy right tail - we do not plot the true posterior here as it is significantly more peaked than the proposal. 

\begin{figure}[!ht]
\centering
 \includegraphics[width=0.97\textwidth]{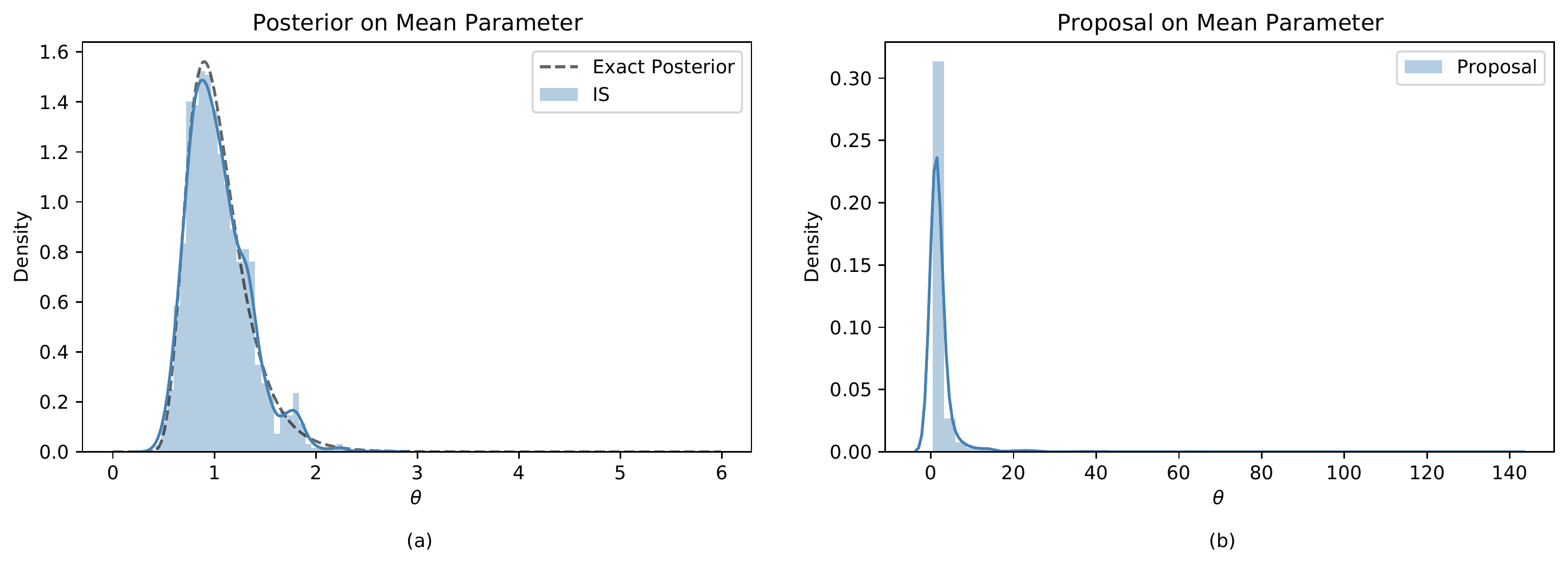}\vspace{-5mm}
  \caption{Random ordering:  (a) Martingale posterior generated via  Algorithm~\ref{alg:joint_sampling} (without \gls*{smc} resampling) and predictive resampling for $\theta$; (b) Proposal distribution before \gls*{IS} reweighting}\label{fig:sim_param_IS}
\end{figure}

\begin{figure}[!ht]
\centering
 \includegraphics[width=0.97\textwidth]{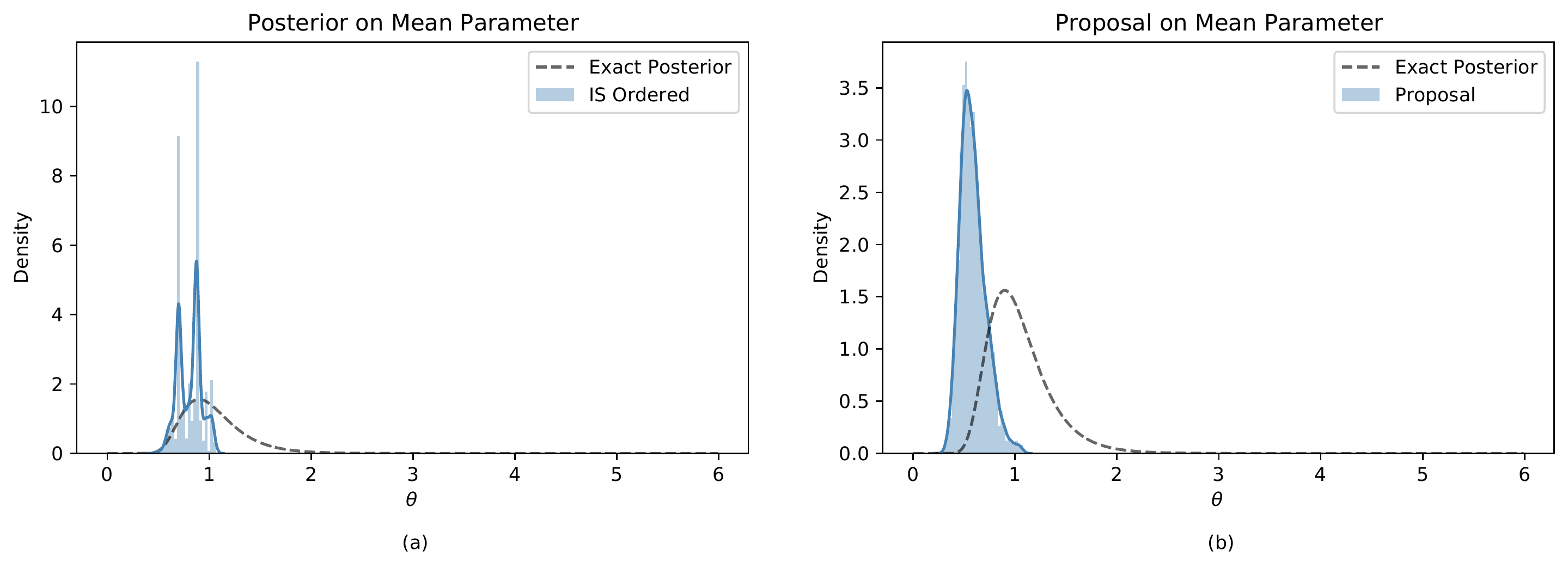}
 \vspace{-5mm}
   \caption{Uncensored/censored ordering:  (a) Martingale posterior generated via  Algorithm~\ref{alg:joint_sampling} (without \gls*{smc} resampling) and predictive resampling for $\theta$; (b) Proposal distribution before \gls*{IS} reweighting}\label{fig:sim_param_IS_ord} \vspace{-2mm}
\end{figure}

\newpage

\subsection{Derivation of Algorithm 1 }\label{sec:supp_IS_deriv}

As discussed in the main paper and above, the \gls*{IS} weights take on a slightly different form under random ordering, which we now derive. For a dataset $\mathcal{D}_n$, denote the indices of observed data points as $\mathcal{I}_{o}$ and censored data points as $\mathcal{I}_c$, so that $\mathcal{I}_o \cup \mathcal{I}_c = \lbrace 1, \dots, n \rbrace$ and $\mathcal{I}_o \cap \mathcal{I}_c = \O$, and $\mathcal{D}_n = \{y_{\mathcal{I}_o}, Y_{\mathcal{I}_c}\geq c_{\mathcal{I}_c}\}$. Our sequential imputation scheme in Algorithm 1 gives the proposal density 
$$
q(y_{\mathcal{I}_c}) = \prod_{i \in \mathcal{I}_c} p(y_i \mid  Y_i \geq c_i, y_{1:i-1}).
$$
Our target however is the conditional density
$$
p(y_{\mathcal{I}_c} \mid Y_{\mathcal{I}_c} \geq c_{\mathcal{I}_c}, y_{\mathcal{I}_o}) \propto {p(y_{\mathcal{I}_c},  Y_{\mathcal{I}_c} \geq c_{\mathcal{I}_c},y_{\mathcal{I}_o})}.
$$
We can factorize the mixed joint density into
\begin{equation*}
   \begin{aligned}
   {p(y_{\mathcal{I}_c},  Y_{\mathcal{I}_c} \geq c_{\mathcal{I}_c},y_{\mathcal{I}_o})} 
   &= \prod_{i \in \mathcal{I}_c} {p(y_i, Y_i \geq c_i \mid y_{1:i-1})} \, \prod_{j \in \mathcal{I}_o} p(y_j \mid y_{1:j-1}) \\
   &= \prod_{i \in \mathcal{I}_c} p(y_i \mid  Y_i \geq c_i, y_{1:i-1}) \, P(Y_i \geq c_i \mid y_{1:i-1})\, \prod_{j \in \mathcal{I}_o} p(y_j \mid y_{1:j-1}) \\
   &= q(y_{\mathcal{I}_c}) \, \prod_{i \in \mathcal{I}_c} P(Y_i \geq c_i \mid y_{1:i-1})\, \prod_{j \in \mathcal{I}_o} p(y_j \mid y_{1:j-1}).
   \end{aligned} 
\end{equation*}
Dividing the above by $q(y_{\mathcal{I}_c})$ gives us the unnormalized importance weights
\begin{equation}\label{eq:supp_IS_weights}
\begin{aligned}
w &=  \prod_{i \in \mathcal{I}_c} P(Y_i \geq c_i \mid y_{1:i-1})\, \prod_{j \in \mathcal{I}_o} p(y_j \mid y_{1:j-1})
\\&= \prod_{i=1}^n \left[\delta_i\,p(y_i \mid y_{1:i-1}) + (1-\delta_i)\,P(Y_i \geq c_i \mid y_{1:i-1})  \right]
\end{aligned}
\end{equation}
where $\delta_i = 1$ if $i \in \mathcal{I}_o$ and $\delta_i =0 $ if $i \in \mathcal{I}_c$. The above is exactly the importance weight in Algorithm 1.

\subsection{Selecting Hyperparameters}\label{sec:supp_hyperparam}

In the copula densities above, we set the bandwidth parameters $a$ or $\rho$ by maximizing the (mixed) joint `marginal likelihood' $p(\mathcal{D}_n)$. Assuming no \gls*{smc} resampling steps for now, we can estimate $p(\mathcal{D}_n)$ through \gls*{IS}:  
$$
\widehat{p}(\mathcal{D}_n) = \sum_{j=1}^B w^{(j)}
$$
where $w^{(j)}$ are the unnormalized \gls*{IS} weights from (\ref{eq:supp_IS_weights}). This can be shown to be a valid estimate, as we are approximating the expectation
\begin{equation*}
    \begin{aligned}
    \int \prod_{i \in \mathcal{I}_c} P(Y_i \geq c_i \mid y_{1:i-1})\, \prod_{j \in \mathcal{I}_o} p(y_j \mid y_{1:j-1}) \, q(y_{\mathcal{I}_c}) \, dy_{\mathcal{I}_c} &= \int p(y_{\mathcal{I}_c},  Y_{\mathcal{I}_c} \geq c_{\mathcal{I}_c},y_{\mathcal{I}_o}) \, dy_{\mathcal{I}_c}\\
    &= {p(Y_{\mathcal{I}_c} \geq c_{\mathcal{I}_c},y_{\mathcal{I}_o})},
    \end{aligned}
\end{equation*}
which is exactly $p(\mathcal{D}_n)$. When \gls*{smc} resampling steps are present, we can then approximate the ratio $Z_i /Z_{i-1}$ at each time step and compute the product to get $Z_n$ (\cite[Section~3.5]{Doucet2009}). Here $Z_i = p(\mathcal{D}_i)$, where $\mathcal{D}_i$ are the data points up to and including datum $i$. 

To maximize $\widehat{p}(\mathcal{D}_n)$, we optimize across a pre-specified grid of hyperparameter values. We note that, although a gradient-based approach may be possible, it is likely to be slow due to the large number of particles $B$ and potentially unstable due to the need to differentiate through an \gls*{IS} estimate. 

\subsection{Initializaton and Standardization}\label{sec:supp_init_standard}

As motivated by the copula derivations above, we initialize the exponential copula update with ${p_0(y) = \text{Lomax}(y \mid a,1)}$, where $a$ is the same hyperparameter as the bandwidth. We set $a$ in an adaptive way, as a default value is difficult to set -  we cannot gauge the tail behaviour from the observed sample in the presence of right-censoring. An equivalent argument applies for $\rho$ in the log-normal copula update, which controls the variance of $p_0 = \text{Log-normal}(y \mid 0, 1/(1-\rho))$.

However, to prevent a wildly inappropriate $p_0$, we opt to normalize the observed times in a heuristic manner to ensure the times are of the right order of magnitude. To illustrate this, we briefly use the alternative notation of observed times and censoring indicators $\{t_i,\delta_i\}_{i = 1:n}$ for convenience, which corresponds one-to-one to the notation $\mathcal{D}_n = \{y_{\mathcal{I}_o}, Y_{\mathcal{I}_c}\geq c_{\mathcal{I}_c}\}$. From the Bayesian model in (\ref{eq:supp_exp_likelihood}), we see that the prior expected value of $\theta$ is $a/b$. The hyperparameter $a$ is a prior pseudo-count, so  we aim for a default target value of $a \approx 1$. As we set $b = 1$, this suggests that we are aiming for a target $E[\theta] \approx 1$ under the exponential model. Finally, we highlight that the \gls*{mle} of the rate $\theta$  for the exponential sampling density takes the form
\begin{equation*}
\widehat{\theta} = \frac{\sum_{i=1}^n \delta_i}{\sum_{i=1}^n t_i}\cdot
\end{equation*}
We thus opt to multiply the times $t_{1:n}$ by $\widehat{\theta}$ above  to ensure an \gls*{mle} of $\theta$ of 1. In the log-normal case, the \gls*{mle} is not tractable unlike in the exponential case. As a result, we prefer to also multiply by $\widehat{\theta}$ above in this case which works well in practice. 

\subsection{Diagnostics}

We now briefly discuss the assessment of the computational accuracy of our method, as we are rely on Monte Carlo and truncation approximations. We provide these diagnostics for our experiments in Section \ref{sec:supp_results}. For the simulation of $y_{\mathcal{I}_c}$, we report the usual diagnostics for \gls*{smc} - that is, we plot the \gls*{ess} as computed in (\ref{eq:supp_ESS}) against time in order to observe the number of resample steps. We also track the number of unique particles of ${y}_{\min(\mathcal{I}_c)}$ as a measure of particle degeneracy. See \citet{Doucet2009} for more details. 

To diagnose the convergence of predictive resampling in the nonparametric case, we track the $L_1$ distance between the starting \gls*{cdf} $P_n$ and forward simulated \gls*{cdf} $P_N$, i.e. we compute $\int |P_N(y) - P_n(y) | \, dy$ which we can approximate numerically. Note that this is the 1-Wasserstein metric. We use this as the survival function is of primary interest; we expect the Wasserstein-1 distance to converge to a constant as $N\to \infty$. In the parametric case, we simply observe the value of $\bar{\theta}_N$, which will also converge to a constant from Doob's result.

\subsection{Copula Regression}
For conditional density estimation, the copula update takes on the form 
\begin{equation}\label{eq:supp_conditreg_conditional}
\begin{aligned}
p_{i+1}(y \mid \mathbf{x}) =\{1-\alpha_{i+1}(\mathbf{x},\mathbf{x}_{i+1}) ~ + \alpha_{i+1}(\mathbf{x},\mathbf{x}_{i+1})\, c_{\rho}\left(q_i,r_i\right)\}\, p_i(y\mid \mathbf{x}),
\end{aligned}
\end{equation}
where $q_i = P_i(y \mid \mathbf{x}), \quad r_i = P_i(y_i \mid \mathbf{x}_i)$. The update above corresponds to the conditional density update in the multivariate copula update \citep{Fong2021}. A simplification is also suggested in \citet{Fong2021} for the form of $\alpha_{i+1}(\mathbf{x},\mathbf{x}_{i+1})$, which is
 \begin{equation}\label{eq:supp_condit_alpha}
\alpha_{i}(\mathbf{x},\mathbf{x}') = \frac{\alpha_{i}\prod_{j=1}^d  c_{\rho_{x}} \left\{\Phi\left(x^j\right),\Phi\left(x'^j\right)\right\}}{1- \alpha_{i} + \alpha_{i}\prod_{j=1}^d   c_{\rho_{x}} \left\{\Phi\left(x^j\right),\Phi\left(x'^j\right)\right\}}.
\end{equation}
Here, $c_\rho$ is the Gaussian copula density as in (\ref{eq:supp_gaussian_copula}). We initialize  $p_0(y \mid \mathbf{x})  = p_0(y)$ which may be the Lomax or log-normal density as described above, independent of $\mathbf{x}$. We also standardize the data in the same way as in Section \ref{sec:supp_init_standard}, ignoring the covariates. For predictive resampling, as mentioned in the main paper, we draw $X_{n+1:N}$ through the Bayesian bootstrap. 
\vspace{5mm}
\newpage
\section{EXPERIMENTS \label{sec:supp_results} }

In this section, we provide further details on each of the individual experiments. 

\subsection{Simulated Data}
For the simulated data example, the aim was to show the equivalence between predictive resampling and posterior sampling in the \textit{parametric} case. Our (well-specified) model is 
 \begin{equation*}
    f_\theta(y) = \frac{\exp\left(- y/\theta \right)}{\theta}, \quad \pi(\theta) = \text{Inverse-gamma}(\theta \mid a_0,b_0)
 \end{equation*} 
where we have reparametrized so that $\theta$ is the mean of the population. Once again, it is convenient to use the $\{t_i,\delta_i\}_{i = 1:n}$ notation for the observed data. Under non-informative censoring, the posterior is simply 
$\pi(\theta \mid \mathcal{D}_n) = \text{IG}(a_n, b_n)$, where
$
a_n = a_0 + \sum_{i=1}^n \delta_i, \quad b_n = b_0 + \sum_{i=1}^n t_i.
$
The posterior predictive is also analytically tractable as the $ \text{Lomax}(a_n,b_n)$ distribution, with density and \gls*{cdf} given in  (\ref{eq:supp_lomax_pdf}) and (\ref{eq:supp_lomax_cdf}). It is also helpful to derive the marginal likelihood, which takes on the form
\begin{equation*}
    \begin{aligned}
    p(\mathcal{D}_n) &= \int \prod_{i = 1}^n [\delta_i \, f_\theta(y_i)\, + (1-\delta_i) \, \bar{F}_\theta(c_i)] \, \pi(\theta) \, d\theta\\
    &= \frac{b_0^{-k}\,\Gamma(k + a_0)}{\Gamma(a_0)}{\left(1+ \frac{\sum_{i=1}^n t_i}{b_0}  \right)^{-(k + a_0)}}\\
    &= \frac{\Gamma(k+a_0)}{\Gamma(a_0)} \frac{b_0^{a_0}}{(b_0 + \sum_{i=1}^n t_i )^{k+a_0}}
    \end{aligned}
\end{equation*}
where $k = \sum_{i=1}^n \delta_i$. Setting $b_0 = 1$, we maximize the above using gradient descent to elicit $a_0$, yielding $a_0 = 1.46$ for our particular example. 

In Algorithm \ref{alg:joint_sampling}, for a censored datum $Y_i \geq c_i$, we wish to simulate  $Y_i \sim p_{i-1}^{c_i}$, where $p_{i-1}^{c_i}(y) = p(y \mid Y_i \geq c_i, y_{1:i-1})$. Once again, we work in the space of \gls*{cdf}s and draw
\begin{equation*}
    \begin{aligned}
U_{i} &\sim \mathcal{U}[P_{i-1}(c_{i}),1 ], \quad Y_{i} = P_{i-1}^{-1}(U_{i}).
\end{aligned}
\end{equation*}
In this case, we require $P_{i-1}^{-1}$, which is tractable and easy to compute, as given in  (\ref{eq:supp_lomax_cdf}). Updating the predictive then involves computing $a_i = a_{i-1}+ 1$ and $b_i = b_{i-1} + Y_i$, and the \gls*{IS} weight update involves the $\text{Lomax}(a_{i-1},b_{i-1})$ \gls*{cdf} at $c_i$.

To predictively resample $Y_{n+1:N}$, we draw
\begin{equation*}
    \begin{aligned}
U_{i} &\sim \mathcal{U}[0,1 ], \quad Y_{i} = P_{i-1}^{-1}(U_{i}),
\end{aligned}
\end{equation*}
followed by the same updates for $a_i$ and $b_i$. For the limiting parameter estimate, we utilize the posterior mean, which takes the form 
$$
\bar{\theta}_N = \frac{b_N}{a_N - 1}
$$
for the $\text{Inverse-gamma}(a_N,b_N)$ posterior. As an aside, note that $\bar{\theta}_N = \sum_{i=1}^N Y_i /N$ would also work as it is a strongly consistent estimator.

In Figure \ref{fig:sim_diagnostics}a, we plot the \gls*{ess} as Algorithm \ref{alg:joint_sampling} progresses. Although there are two resampling steps, resulting in a decrease in the number of unique particles each time, we still have approximately 600 particles at the end, which is sufficient for estimating $p(y_{n+1} \mid \mathcal{D}_n)$ accurately. In Figure \ref{fig:sim_diagnostics}b, we plot the paths of $\bar{\theta}_i$ for a few predictive resampling chains, where we see that $N = 2000 + n$ is sufficient for convergence.

\begin{figure}[!ht]
\centering
 \includegraphics[width=0.97\textwidth]{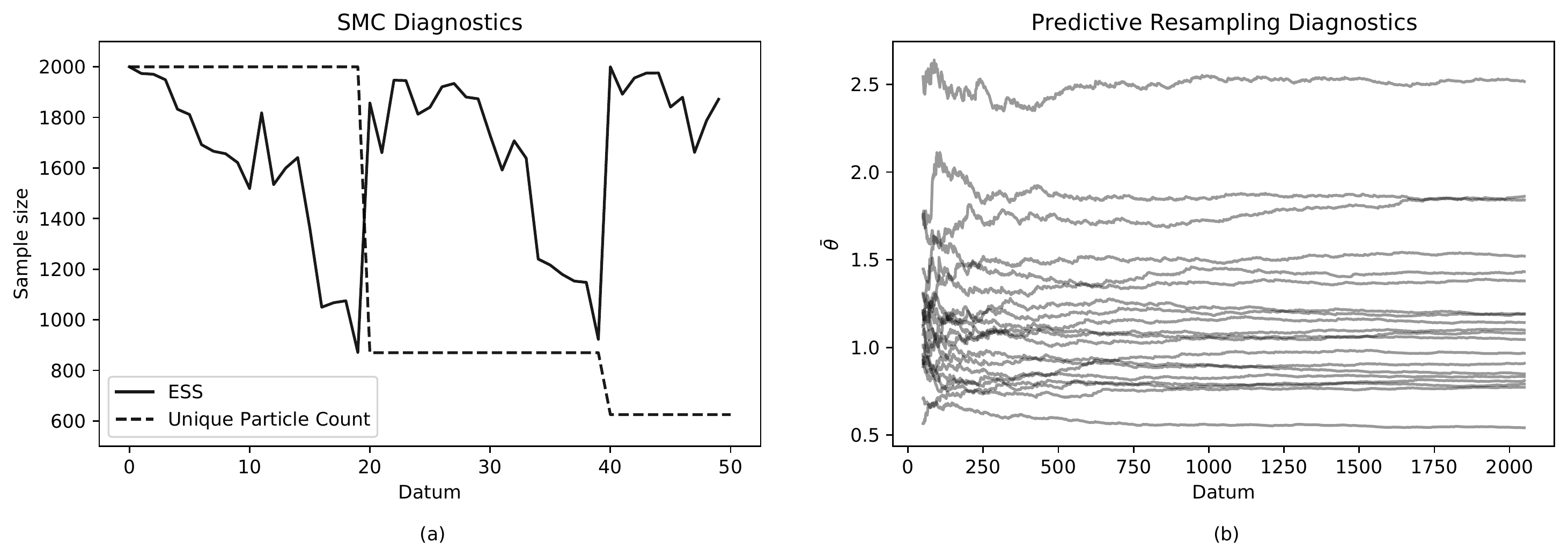}\vspace{-5mm}
  \caption{(a) \gls*{ess} and unique particles; (b) Trajectories of $\bar{\theta}$ from predictive resampling}\label{fig:sim_diagnostics}
\end{figure}
 \newpage
\subsection{Primary Biliary Cirrhosis}
For the copula method fit to the full dataset, we set the bandwidth $a$ by maximizing $\widehat{p}(\mathcal{D}_n)$ on the grid of $[0.5,0.6,0.7,0.8,0.9]$ and $[1.1,1.2,1.3,1.4,1.5]$ for the treatment and placebo datasets respectively. This gives $a = 0.8$ and $a = 1.2$ for the treatment and placebo respectively. For the cross-validation runs with the 50-50 train-test split, we however use the grid $[1.1,1.2,1.3,1.4,1.5]$ for both treatment and placebo; the difference in grid for placebo is due to the 50-50 train-split preferring a different value of $a$. For the baseline, we elicit a \gls*{dpmm} with the exponential kernel and Gamma prior, that is the \gls*{dpmm} with kernel and centering measure given in (\ref{eq:supp_exp_likelihood}) with $b = 1$. We fit the \gls*{dpmm} using the \texttt{R} package \texttt{dirichletprocess} \citep{ross2018}. Like the copula method, we set $a = 0.8$ and $a = 1.2$ in the Gamma centering measure of the \gls*{dpmm} for the treatment and placebo datasets respectively for both the full and the cross-validation fits. 

For computing Figure \ref{fig:pbc2_posterior} in the main paper, predictive resampling was carried out on a grid of size 149 (not 100 as incorrectly stated in the main paper) between 0 and 21 years. In Figure \ref{fig:pbc2_diagnostics}a, the \gls*{ess}/particle count plots show that no resampling steps were required.  In Figure \ref{fig:pbc2_diagnostics}b, the 1-Wasserstein distance between $P_n$ and $P_N$ has roughly converged at $N = 2000 + n$ forward steps as implemented in the paper.

For reference, we also plot the posterior mean and 95\% credible intervals of the random limiting density $p_N$ for the copula method and the \gls*{dpmm} in Figure \ref{fig:pbc2_pdf}. We see that while the posterior means are similar, the uncertainty bands are noticeably different.

\begin{figure}[!ht]
\centering
 \includegraphics[width=0.97\textwidth]{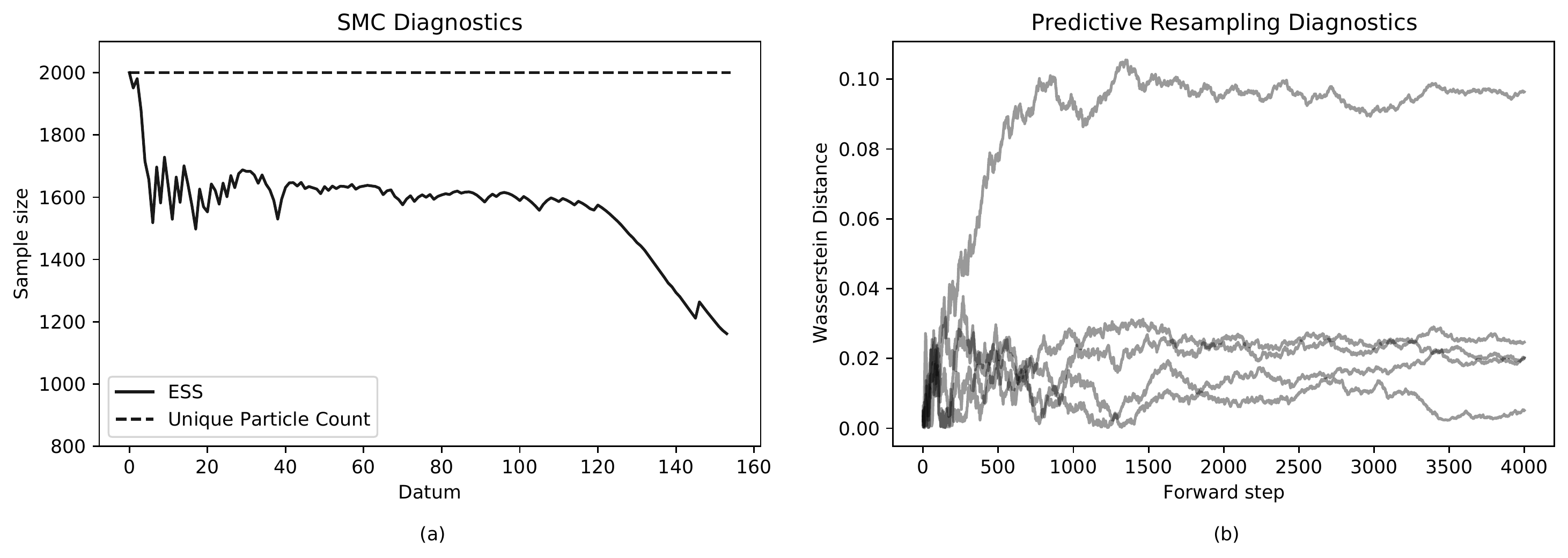}\vspace{-5mm}
  \caption{ (a) \gls*{ess} and unique particles; (b) Trajectories of Wasserstein-1 Distance from predictive resampling}\label{fig:pbc2_diagnostics}
\end{figure}

\begin{figure}[!ht]
\centering
 \includegraphics[width=0.97\textwidth]{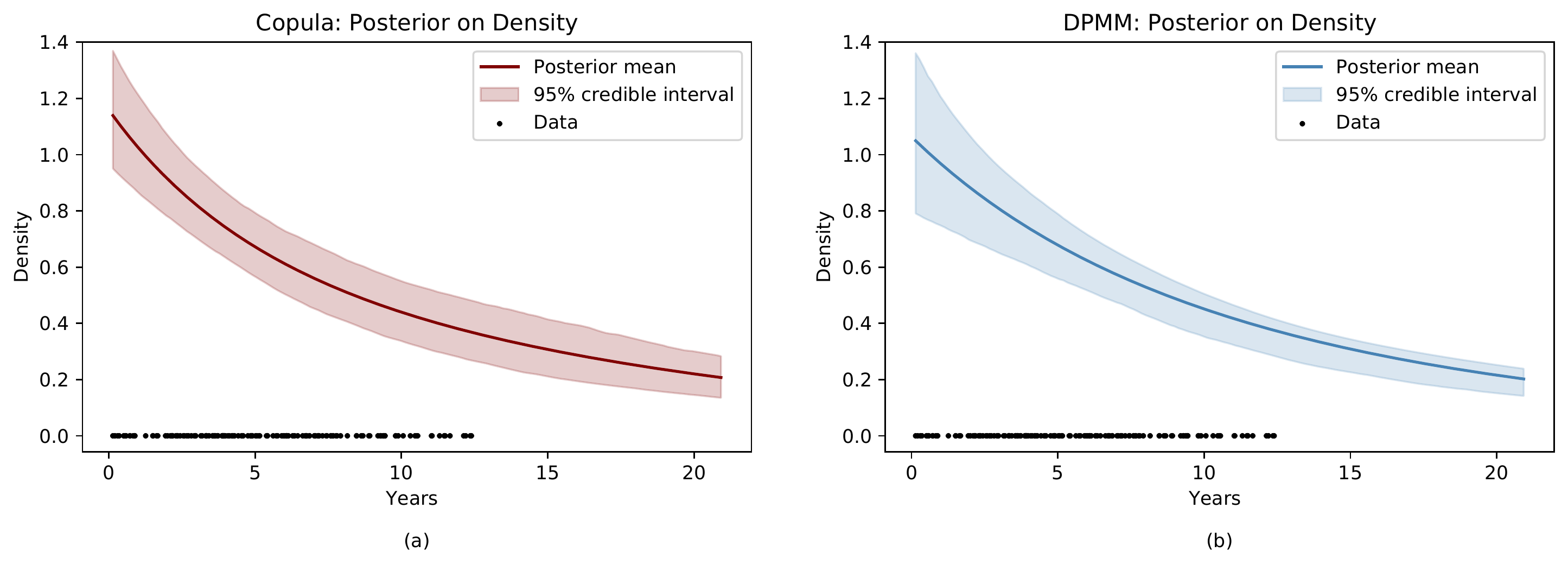}\vspace{-5mm}
  \caption{ Posterior mean and 95\% credible interval of density for (a) copula method and (b) \gls*{dpmm}.}\label{fig:pbc2_pdf}
\end{figure}

\subsection{Survival Regression}

For both the melanoma and kidney examples, we set the bandwidth $a$ on the grid [0.5,0.6,0.7,0.8,0.9] for the copula method. For the nonparametric baseline, we use the \texttt{ddpsurvival} function with default settings from the \texttt{ddpanova} package \citep{deIorio2004}. The package uses the methodology introduced in \citet{deIorio2009} which extends the ANOVA dependent Dirichlet process (DDP) of \citet{deIorio2004}. This method however employs a Gaussian kernel DDP, so we log transform the survival times before fitting, which is equivalent to running a DDP with a log-normal kernel. We also compare to the linear AFT with the log-normal distribution, using the \texttt{survreg} function in the \texttt{survival} package  \citep{survival-package} in \texttt{R}. 

For the melanoma dataset, we also provide additional results for fitting to the full dataset. For the copula method, optimizing the hyperparameters and fitting the model gives us $\rho = 0.6,\rho_x = 0.8$. For Figure \ref{fig:melanoma_surv_KM} in the main paper, we compute the predictive density on a grid of size 56 between 0 and 5565 days (largest datum).
Similarly, Figure \ref{fig:melanoma_median_fun}a shows the equivalent plot for the DDP - the fit with the \gls*{km} plots for the DDP is not as close as the copula method. We also compute the median survival time as a function of $x$ on a $x$-grid of size $40$, which is shown in Figure \ref{fig:melanoma_median_fun}b. The median function of the DDP is smoother than that of the copula method, where the latter is controlled by the value of $\rho_x$.

To demonstrate predictive resampling, we consider the conditional density/survival function at $x = 3.4$. For the copula method, we carry out $N = 10000 + n$ forward samples with $B = 2000$, which takes 5.6s. However, we point out that this needs to be run for each $
\mathbf{x}$ of interest, which may be costly. We plot the posterior mean and 95\% credible intervals for the copula estimate of the conditional density and survival function in Figure \ref{fig:melanoma_dens_surv}, with the DDP posterior mean functions overlaid. Finally, in Figure \ref{fig:melanoma_diagnostics}a, we see that 3 resampling steps reduce the unique particle count to $\approx 400$, and Figure \ref{fig:melanoma_diagnostics}b demonstrates that $ N = 10000 + n$ is sufficient for convergence.

\begin{figure}[!ht]
\centering
 \includegraphics[width=0.97\textwidth]{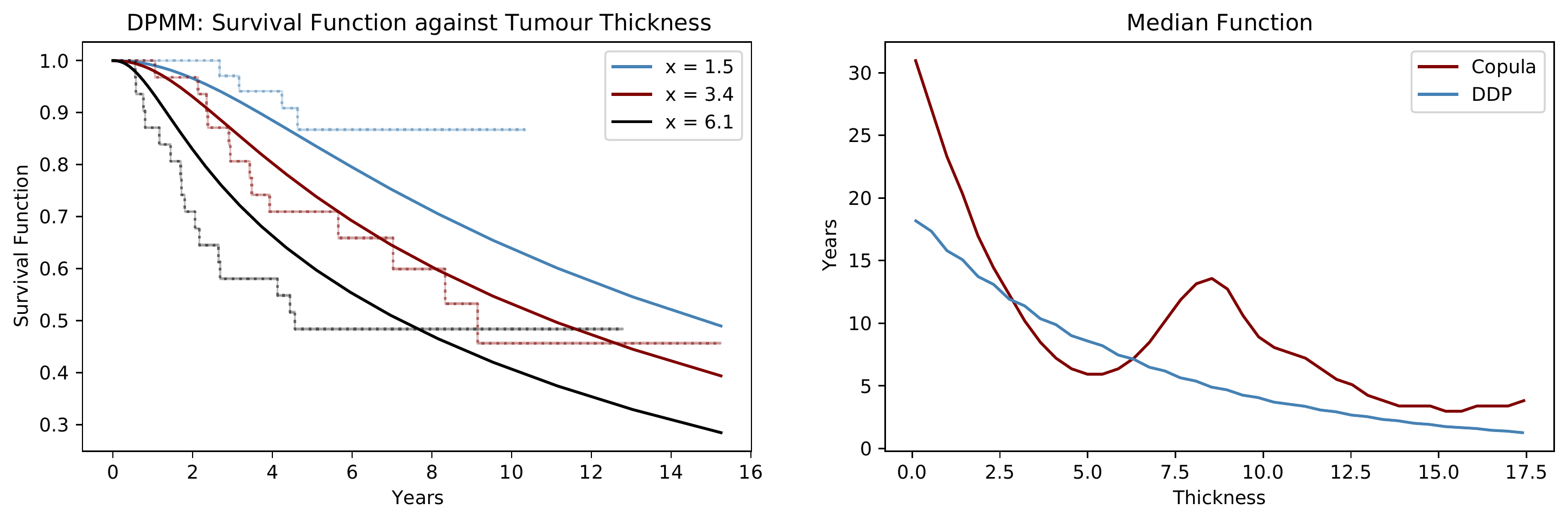}\vspace{-5mm}
  \caption{(a) Survival function for DDP method (\fullmidk) for $x =\{1.5,3.4,6.1\}$ with \gls*{km} (\dottedmidk) fit to windows $\{(1.255,1.75),(2.7,4.1),(4.1,8.1)\}$; (b) Median survival time against tumour thickness $x$}\label{fig:melanoma_median_fun}
\end{figure}

\begin{figure}[!ht]
\centering
 \includegraphics[width=0.97\textwidth]{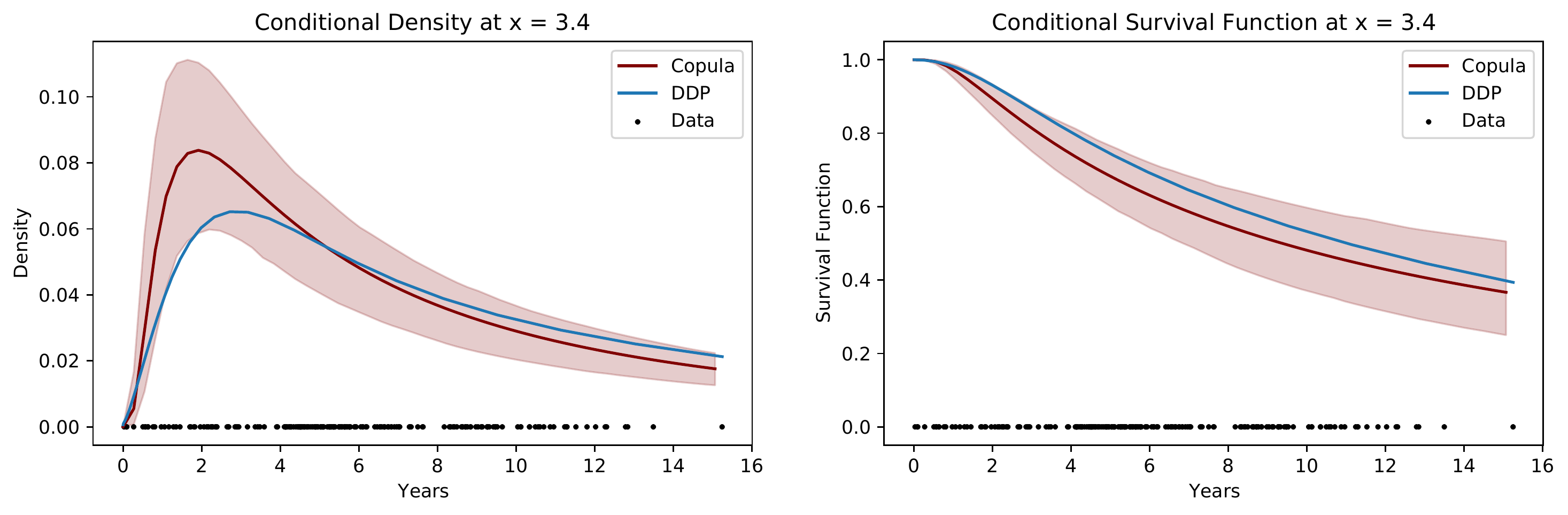}\vspace{-5mm}
  \caption{Posterior mean and 95\% credible interval (for copula method only) of density for (a) density and (b) survival function.}\label{fig:melanoma_dens_surv}
\end{figure}

\begin{figure}[!ht]
\centering
 \includegraphics[width=0.97\textwidth]{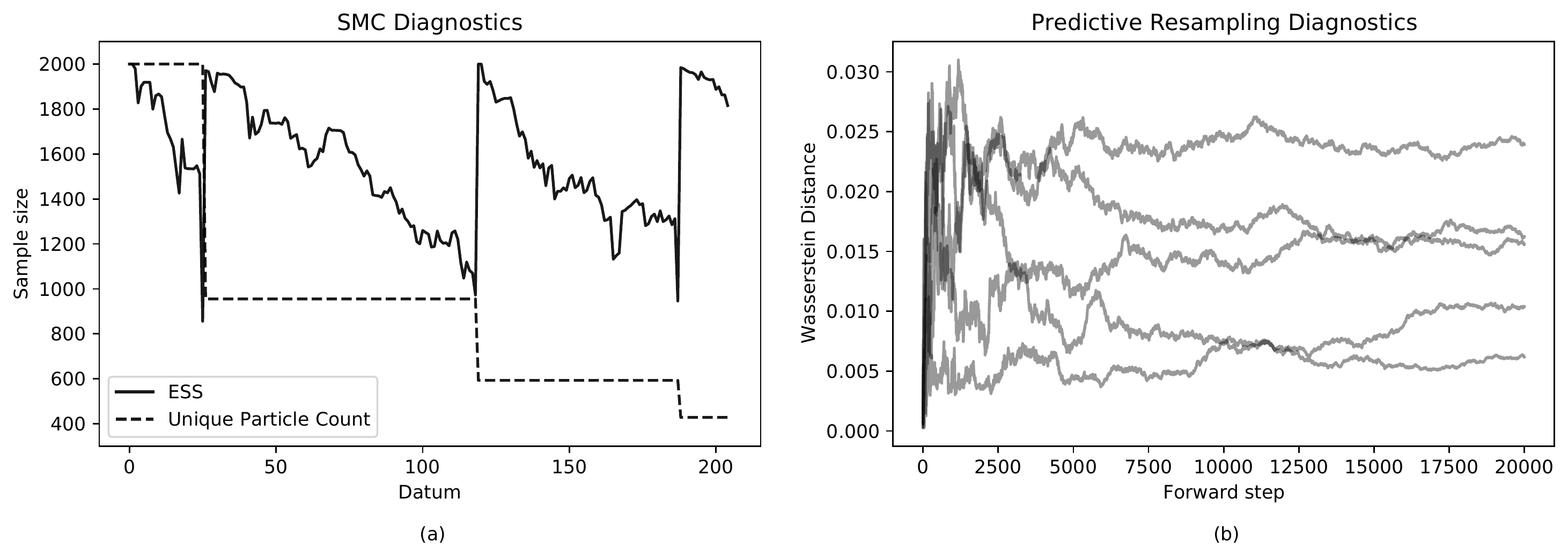}\vspace{-5mm}
  \caption{(a) \gls*{ess} and unique particles; (b) Trajectories of Wasserstein-1 Distance from predictive resampling.}\label{fig:melanoma_diagnostics}
\end{figure}

\end{document}